\documentclass[11pt,a4paper,notitlepage]{article}
\usepackage[utf8]{inputenc}
\usepackage[T1]{fontenc}
\usepackage{amsmath, amsthm, amssymb, amsfonts}
\usepackage{graphicx}
\usepackage{enumerate}
\usepackage[table]{xcolor}
\usepackage{listings}
\usepackage{float}
\usepackage{subcaption}
\usepackage{booktabs}
\usepackage{bm}
\usepackage{placeins}
\usepackage{natbib}
\usepackage{multirow}
\usepackage[draft]{changes}

\setcitestyle{authoryear, open={(},close={)}}

\usepackage{authblk}

\usepackage{color} 
\definecolor{mygreen}{RGB}{28,172,0} 
\definecolor{mylilas}{RGB}{170,55,241}
\definecolor{MyDarkGreen}{rgb}{0.0,0.4,0.0}
\definecolor{Blue}{rgb}{0.0,0,0.4}
 \usepackage{hyperref}
 
 \usepackage{todonotes}

\usepackage{geometry}
\usepackage{fancyhdr}

\newtheorem{theorem}{Theorem}[section]

\newtheorem{definition}[theorem]{Definition}
\newtheorem{proposition}[theorem]{Proposition}
\newtheorem{assumption}[theorem]{Assumption}
\newtheorem{remark}[theorem]{Remark}

\newcommand{\thetab}{\bm{\theta}}
\newcommand{\rb}{\bm{r}}

\newcommand{\R}{\mathbb{R}}
\newcommand{\p}{\mathbb{P}}
\newcommand{\Q}{\mathbb{Q}}

\usepackage{pifont}

\usepackage{soul}

\title{{\huge Locally tail-scale invariant scoring rules for evaluation of extreme value forecasts}}
\author[1]{Helga Kristín Ólafsdóttir}
\author[1]{Holger Rootzén}
\author[2]{David Bolin}
\affil[1]{Department of Mathematical Sciences, Chalmers University of Technology and University of Gothenburg, Gothenburg, Sweden}
\affil[2]{Computer, Electrical and Mathematical Sciences and Engineering Division, King Abdullah University of Science and Technology, Thuwal, Saudi Arabia}
\date{}

\begin{document}
\maketitle

\begin{abstract}


Statistical analysis of extremes can be used to predict the probability of future extreme events, such as large rainfalls or devastating windstorms. The quality of these forecasts can be measured through scoring rules. Locally scale invariant scoring rules give equal importance to the forecasts at different locations regardless of differences in the prediction uncertainty. This is a useful feature when computing average scores but can be an unnecessarily strict requirement when mostly concerned with extremes. We propose the concept of local weight-scale invariance, describing scoring rules fulfilling local scale invariance in a certain region of interest, and as a special case local tail-scale invariance, for large events. Moreover, a new version of the weighted Continuous Ranked Probability score (wCRPS) called the scaled wCRPS (swCRPS) that possesses this property is developed and studied. The score is a suitable alternative for scoring extreme value models over areas with varying scale of extreme events, and we derive explicit formulas of the score for the Generalised Extreme Value distribution. The scoring rules are compared through simulation, and their usage is illustrated in modelling of extreme water levels, annual maximum rainfalls, and in an application to non-extreme forecast for the prediction of air pollution.

\end{abstract}

\section{Introduction}

The main aim for statistical analysis of extremes in fields such as hydrology, meteorology, climate science, and finance is  forecasting, or prediction, of the risk of occurrence of future dangerous extreme events, such as huge rainfall events or very large fluctuations of prices of a financial portfolio. 
Typically, the forecast is in terms of an estimated distribution, for example the estimated distribution of the size of the largest event which will occur during some specified period of time. In this case, the forecast is said to be probabilistic. The forecast distribution is then often presented to decision makers in a reduced form, perhaps as a predicted probability of exceeding one, or several, high thresholds. In many cases, there are several competing forecast distributions, and a standard way of choosing between these competing forecasts is to compute goodness of fit measures of the distributions based on the existing data, and then use the distribution that fits best. This approach is, for example,  broadly used in hydrology \citep{Cugerone2015JohnsonDistribution}. However, an alternative gaining increasing interest is to base the selection on proper scoring rules \citep{Zamo2018}. 

A scoring rule is a functional $S(\p,y)$ that takes in an outcome $y$ and a probabilistic forecast $\p$, the predictive distribution for the outcome. If the outcome $Y\sim\Q$ is a random variable with distribution $\Q$, the score $S(\p,Y)$ is also a random variable, with expected value denoted by $S(\p,\Q)$. To be able to make an earnest prediction, it is preferred that the expected score is maximised under the correct model, i.e., that $S(\Q,\Q)\geq S(\p,\Q)$, for all predictive distributions $\p$. In this case the score is said to be proper \citep{Gneiting2007}. If strict inequality holds for all $\p \neq \Q$, the scoring rule is  strictly proper.

An example of a commonly used proper scoring rule is the Continuous Ranked Probability Score (CRPS), which is defined as
\begin{equation*}
CRPS(\p,y)=-\int_{\R} (F(x)-1(x\geq y))^2 dx, 
\end{equation*}
where $F$ is the cumulative distribution function of the forecast distribution $\p$, and $1(\cdot)$ is the indicator function. Another popular example is the logarithmic score $LS(\p,y) =\log f(y)$, where $f$ is the density function of $\p$. Both of these scores are widely used in meteorology, climate science, and finance ~\citep{Opschoor2017CombiningRules,Ingebrigtsen2015EstimationField,Haiden2019EvaluationUpgrade}.
The choice of scoring rule allows an evaluator to emphasise the desired features of the prediction. For extreme events, one typically is interested in the behaviour of the prediction in the tail of the distribution. One way of doing this is to simply multiply the proper score $S_0$ with a weight function $w$ and form a new score $S(\p,y)=w(y)S_0(\p,y)$. However, this score is not proper unless $w(y)$ is constant~\citep{Gneiting2011ComparingRules} and alternative ways of weighting should be used. Several such proper weighted scores exist~\citep[see e.g.][]{Diks2011,Todter2012GeneralizationDecomposition,Gneiting2011ComparingRules}, with a popular one being the threshold weighted CRPS (wCRPS) by \citet{Gneiting2011ComparingRules}, where a weight function is included in the kernel of the CRPS,
\begin{equation*}
    wCRPS(\p,y)=-\int_{\R} w(x)(F(x)-1(x\geq y))^2dx.
\end{equation*}
Here $w(x)$ is a non-negative function that can be chosen to put more emphasis on large values. However, what weight function to use is not obvious. The forecaster's dilemma, as presented by \citet{Lerch2017}, describes the problem that forecasters can be encouraged to exaggerate their predictions, since the effect of missing a calamity is worse than wrongfully predicting one. To avoid this, the choice of the weight function should ideally be made by the user of the forecasts rather than by the forecaster. 

In the case of multiple observations, ${\bm y}=(y_1,,,, y_n)$, and a forecast $\p=(\p_1,...,\p_n)$ the joint score of the forecast is defined as the average score,
\begin{equation*}
    S(\p,{\bm y})=\frac{1}{n}\sum_{i=1}^n S(\p_i,y_i),
\end{equation*}
and this is used to inform the selection of prediction method.

As recently shown by \citet{Bolin2022LocalRulesCustom}, many common scoring rules, such as the CRPS, are ``scale dependent'' in the sense that they put higher importance on forecasts with higher prediction uncertainty in the case when the observations $y_i$ in the average score have varying predictability. This is, sometimes, an undesirable property. For example, if one wants to predict extreme rainfalls in some spatial area  where the scale of the events varies, infrastructure in the parts where the scale of events is small and in parts where the scale is large will have adapted to these differences, and accuracy is equally important throughout the area.  \citet{Bolin2022LocalRulesCustom} introduced the concept of ``local scale invariance'' to describe scoring rules which do not have this feature, and showed that the logarithmic score is an example of a score that is both strictly proper and locally scale invariant. However, use of the logarithmic score requires computation of the probability density function of the predictive distribution, which is not always easy to compute, or might not even exist.

Even though local scale invariance is an important property, it might be an unnecessarily strict requirement if one is mostly concerned with extremes. In this work, we therefore introduce a new concept of ``local weight-scale invariance'' to describe scoring rules which have the property of local scale invariance for a certain region of interest. As a special case, when the region of interest is large events, we obtain the concept of ``local tail-scale invariance''. We develop and study a new version of the wCRPS called the scaled wCRPS, or swCRPS, which is locally weight-scale invariant. The Generalised Extreme Value (GEV) distribution is of special interest in extreme value statistics, and we derive explicit formulas of swCRPS for the GEV distribution. These scores are analysed in terms of local scale invariance. We
compare the swCRPS and the censored likelihood score to the more common logarithmic score, the CRPS, and the wCRPS in a simulation study and in case studies where different models for water levels in the Great Lakes, extreme rainfall events in the Northeastern U.S., and air pollution in the Piemonte region in Italy are evaluated using the different scoring rules.

The article is structured as follows: Section \ref{sec:preliminaries} gives a background on scoring rules, weighting of scores,  and local scale independence of scores. In Section \ref{sec:scaled:weghted:scoring:rules}, scale independence of weighted scores is discussed. Section \ref{sec:results:simulation} contains simulation studies, and Section \ref{sec:results:casestudy} presents case studies. Finally, discussion and conclusions can be found in Section \ref{sec:conclusion}. Derivation of explicit formulas for swCRPS for the GEV distribution, a background on modelling and scoring of extremes,  and proofs of propositions
are included in the appendices of the article.
\section{Background}\label{sec:preliminaries}

In this section we provide more background information about the CRPS and its weighted variants, provide details about scale dependence of scoring rules, and introduce some commonly used models for annual flow and precipitation maxima. In the following, $E_{\p}[g(X)]$ denotes the expectation of $g(X)$ when $X\sim\p$ is a random variable with distribution $\p$, and $E_{\p,\Q}[g(X,Y)]$ denotes the expectation of $g(X,Y)$ for independent random variables $X\sim\p$ and $Y\sim \Q$.

\subsection{The CRPS and weighted scoring rules for extremes}
A different representation of CRPS,  obtained by \citet{Baringhaus2004}, is
\begin{equation}\label{eq:neg:crps}
    CRPS(\p,y)= \frac{1}{2}E_{\p,\p}[|X-X'|]-E_\p[|X-y|],
\end{equation}
where $X,X'$ are independent random variables with the same distribution $\p$ and finite first moment. This expression can be derived by noting that  $|X-Y|=\int_\R 1\{X\leq u\leq Y\}+1\{Y\leq u\leq X\} du$ in Eq. \eqref{eq:neg:crps}, using Fubini's theorem to compute $E|X-Y|$ in terms of the distribution functions $F$ and $G$ of $X$ and $Y$ respectively, and finally setting $G(x)=1\{x\geq y\}$. This representation is useful in cases where a closed form expression for $F$ does not exist, which might be the case when using ensemble forecasts in weather prediction, or when evaluating the predictive performance of complicated hierarchical models. In these cases, the expected values can be approximated through Monte Carlo simulation from the predictive distribution. \citet{Zamo2018} compared different estimators of CRPS and made a recommendation on how to choose the most accurate one based on the ensemble type available.

Similarly for the wCRPS, one can show that
\begin{equation}\label{eq:wCRPS}
    wCRPS(\p,y)= \frac{1}{2}E_{\p,\p}[g_w(X,X')]-E_\p[g_w(X,y)],
\end{equation}
where $g_w(x,x'):=\left|\int_{x'}^x w(t)dt\right|$, provided the expectations above are finite. If one is interested only in the upper-tail behaviour, one might choose the indicator weight function
\begin{equation}\label{eq:qwf}
    w_u(x):=1(x\geq u),
\end{equation}
for some real $u$. 
Instead of choosing the weight function as the indicator function $w_u$, one might want to consider some other variation to incorporate the whole dataset instead of just the extremes. For example, \citet{THORARINSDOTTIR2018155} suggested $w_2(y)=1+1w_u(y)$ and $w_3(y)=1+w_u(y)u$ as alternatives,  
where $u$ was chosen as the $97.5\%$ observed quantile. In an example in Subsection~6.3. of \citet{THORARINSDOTTIR2018155}, the best choice was to use $w_2$ as a weight function in the wCRPS. All these weight functions can be written as $w(y) =a+b w_u(y)$,
where $a,b,u$ are constants, and result in a weighted sum of the wCRPS with threshold weight function and the CRPS.

\citet{TAILLARDAT2022} and \citet{Brehmer2019} show that expected scores are not suitable for scoring of ``max-functionals'', functionals which  are determined by the behaviour of the distribution at $+\infty$, with an example being the so-called Extreme Value Index.  They then argue that this means that use of expected scores may be unsuitable for extremes. However, standard practical use of extreme value statistics is not only concerned with the prediction of tail functionals. If  one is interested in predicting maxima, e.g. yearly maxima of daily rainfall, one obtains a sample of earlier yearly maxima, fits some distribution, perhaps a GEV distribution, to this sample and uses the fitted distribution for prediction. This is completely parallel to fitting a normal distribution to a sample and using the fitted distribution for prediction, and understanding and use of expected scoring rules are the same for the two cases. If one instead is interested in prediction of occurrences and sizes of excesses of a high threshold, one obtains a sample of such excesses and uses the sample for prediction. Again this is completely parallel to using a fitted normal distribution for prediction. Thus, we believe that there is still a value in using expected scoring rules for extremes as long as one is not solely interested in max-functionals such as the extreme value index.

\citet{TAILLARDAT2022} also suggested to treat observed scores as random variables, and to use qq- and pp-plots and Cramer-von Mises statistics to compare scores. This seems like a quite useful idea, both for scoring extremes and for scoring non-extreme values, although it is not entirely clear how to use this to obtain generally applicable rules for ranking models.

\citet{TAILLARDAT2022} and \citet{Brehmer2019} discuss scoring of extreme value tail functionals, i.e.,  functionals that only depend on the behaviour of the distribution at $+\infty$, such as the so-called extreme value index. For a number of examples of scoring rules, they show that to any “true” distribution, there are alternative distributions for which the scoring rules have expectations which are arbitrarily close to the expected score of the true distribution, but dramatically different tail functionals. The intuition, and sometimes the proofs of these results, use the following model for the alternatives,
\[F(x)=(1-\varepsilon) F_1 (x)+\varepsilon F_2 (x),\]
where $\varepsilon$ is a positive number which can be chosen arbitrarily small, $F_1$ is the true distribution function, and $F_2$ has a heavier tail. Then the tail functional of $F$, e.g., the extreme value index, is the same as the that of $F_2$, since $F_2$ dominates at $+\infty$ for any $\varepsilon>0$. But still it is possible to make the expected score for $F$ and $F_1$ arbitrarily close by making $\varepsilon$ small enough.

This fact is used as an argument against using average scoring rules for extremes. However, it should be noted that the same argument can be made against using scoring rules in non-extreme contexts. For example, if one is interested in the variance of the true distribution $F_1$, then for arbitrarily small $\varepsilon$ one can choose $F_2$ so that $F$ has any variance, but the scoring rules give arbitrarily close expected values for $F$ and $F_1$. The only way to avoid this is by mathematical assumption, which for extremes would be to rely on the standard asymptotic theory for how the tails are connected with the observed data. More generally, any scoring rule can only address the parts of the sample space which is explored by the data, and in particular if one has $n$ data points, either any values, or just large values, and $n\varepsilon$ is small, then it is not possible to use the data to distinguish between the models $F$ and $F_1$. This applies to both extremes and ordinary values. 

\citet{TAILLARDAT2022} also suggested to treat observed scores as random variables, and to use qq- and pp-plots and Cramer-von Mises statistics to compare scores. This seems like a useful idea, both for scoring extremes and non-extreme values, although it is not entirely clear how to use it to obtain a generally applicable rule for ranking models.

\subsection{Local scale invariance and kernel scores}
As mentioned above, the CPRS is a scale dependent scoring rule.
\citet{Bolin2022LocalRulesCustom} argued that this can sometimes lead to wrong conclusions, e.g.,  if one uses predictions from  different spatial locations such as in the extreme rainfall analysis discussed in the introduction.
In this section, we make this notion of scale dependence more specific and discuss kernel scores.

For a given probability measure $\Q$, let $\Q_{\thetab}$ for $\thetab=(\mu,\sigma)\in\R\times\R^+$ denote a  location-scale transform of this measure. That is, if $Z \sim \Q$, then $\mu + \sigma Z \sim \Q_{\thetab}$. Note that $\Q_{(0,1)} = \Q$. Now consider a small location-scale misspecification of $\thetab$ that is proportional to $\sigma$, i.e., $\Q_{\thetab+k\sigma \rb}$, where $\rb=(r_1,r_2)$ is a two dimensional unit vector representing the direction of perturbation and $k\in \R$ is a constant. For example, a model with the correct location, $\mu$, but a perturbation of the scale parameter $\sigma+k\sigma$, can be written as $\Q_{\thetab+k\sigma \rb}$ with $\rb=(0,1)$. The idea of local scale invariance is that, for small perturbations, the difference between the score for the perturbed model and the true model  should not depend on the scale. We make this precise in the following definition. 

In the definition and later, for a scoring rule $S$ we use the notation $\mathcal{P}_S$ for the set of probability measures on $(\Omega,\mathcal{F})$ such that if $\p \in \mathcal{P}_S$ and $\Q \in \mathcal{P}_S$ then $|S(\p, \Q)| < \infty$. Further, for  a  set of probability measures
$\mathcal{Q}_0$, we write  $\mathcal{Q} = \{\Q_{\thetab} : \Q \in \mathcal{Q}_0, \thetab \in \mathbb{R}\times \mathbb{R}^+\} $ for the set of probability measures in $\mathcal{P}_S$  which can be obtained as location-scale transforms of measures in $\mathcal{Q}_0$.
If the scoring rule $S$ is proper and twice differentiable with respect to $\theta$, a second order Taylor expansion yields that 
\begin{equation}
\label{eq:scale:invariance} 
S(\Q_{\thetab},\Q_{\thetab})-S(\Q_{\thetab + k\sigma \rb},\Q_{\thetab})= k^2\sigma^2\rb^T s(\Q_{\thetab}) \rb+o(k^2),
\end{equation}
as $k\searrow 0$. The score $S(\Q_{\thetab+k\sigma\rb},\Q_{\thetab})$ has a maximum at $k=0$ due to the properness of $S$. Hence, $\nabla_{\thetab} S(\Q_{\thetab},\Q)\vert_{\Q = \Q_{\thetab}}=0$ and the first order term of the Taylor expansion vanishes. Here, the function $s(\Q_{\thetab}):=\frac{1}{2}\nabla^2_{\thetab} S(\Q_{\thetab},\Q)\vert_{\Q = \Q_{\thetab}}$ is called \emph{scale function} of the scoring rule $S$.
\begin{definition}
\label{def:scale:invariance} 
\citep{Bolin2022LocalRulesCustom} 
Let $S$ be a proper scoring rule with respect to some class of probability measures $\mathcal{P}$ on $(\R,\mathcal{B}(\R))$, and assume that $\mathcal{Q}_0$ is a set of probability measures such that $\mathcal{Q}\subseteq \mathcal{P} \cap \mathcal{P}_S$. 
If $s(\Q_{\thetab})$ exists for all $\Q\in\mathcal{Q}_0$, $\thetab=(\mu,\sigma) \in \R\times \R^+$, and satisfies $s(\Q_{\thetab})\equiv \frac{1}{\sigma^2} s(\Q_{(0,1)})$ we say that $S$ is locally scale invariant on $\mathcal{Q}$.
\end{definition}

 One can show that the CRPS has scale function $s(\Q_{\thetab}) = \sigma s(\Q_{(0,1)})$ for $\thetab = (\mu,\sigma)$, which means that the difference in Eq. \eqref{eq:scale:invariance} scales linearly with $\sigma$. The log-score on the other hand has a scale function $s(\Q_{\thetab}) = \frac{1}{\sigma^2} H_{\Q}$ for a matrix $H_{\Q}$ that is independent of $\thetab$. The factor $\sigma^2$ therefore cancels in Eq. \eqref{eq:scale:invariance}, so that the expression is independent of the scale \citep{Bolin2022LocalRulesCustom}.

A scoring rule that can be written as 
\begin{equation}\label{eq:kernel:first}
S_g^{ker}(\p,y):=\frac{1}{2}E_{\p,\p}[g(X,X')]-E_\p[g(X,y)],    
\end{equation}
where $g$ is a non-negative continuous negative definite kernel is called kernel score~\citep{Dawid2007TheRules}. In order to construct locally scale invariant scoring rules, \citet{Bolin2022LocalRulesCustom} introduced a generalisation of this class of scoring rules, the generalised kernel scores, which are proper scoring rules defined as
\begin{equation}\label{eq:kernel:standard}
  S_g^h(\p,y):=h(E_{\p,\p}[g(X,X')])+2h'(E_{\p,\p}[g(X,X')])(E_\p[g(X,y)]-E_{\p,\p}[g(X,X')]),
\end{equation}
where $h$ is any monotonically decreasing, convex, and differentiable function on $\R^+$. 
For the particular choice $h(x)=-\frac{1}{2}\log(x)$ and $g(x,y)=|x-y|$ in Eq. \eqref{eq:kernel:standard}, one obtains the scoring rule $SCRPS(\p,y)+1$,  where SCRPS is the scaled CRPS scoring rule
\begin{equation}\label{eq:scrps}
    SCRPS(\p,y):=-\frac{E_\p[|X-y|]}{E_{\p,\p}[|X-X'|]}-\frac{1}{2}\log(E_{\p,\p}[|X-X'|]),
\end{equation}
which was shown to be locally scale invariant in \citet{Bolin2022LocalRulesCustom}. Another way of deriving this scoring rule is to use the fact that for any negative proper scoring rule $S(\p,x)$ on 
a set of probability measures $\mathcal{P}$, the transformed score 
\begin{equation}
  S_S^{trans}(\p,y):=\frac{S(\p,y)}{|S(\p,\p)|}-\log(|S(\p,\p)|)\label{eq:scale-invariant}  
\end{equation}
is also a proper scoring rule on $\mathcal{P}$ \citep{Bolin2022LocalRulesCustom}.

\section{Scaled weighted scoring rules}\label{sec:scaled:weghted:scoring:rules}

In this section, we introduce the concept of local weight- and tail-scale invariance as less restrictive alternatives to local scale invariance. We then survey and discuss the use of the scaled wCRPS to score extremes and derive its scaling properties. In the final subsection, we propose alternative combined scoring rules for extremes, concentrating on a scoring rule based on the logarithmic score.

\subsection{Local weight-scale invariance}
When scoring extremes, we are mostly interested in the tail properties of the distributions, and in other cases the interest might be on some specific range of values which not necessarily is extreme. Therefore, it might not be all that important to have full local scale invariance as long as this property  holds in the region of interest. To make this precise, recall that $\mathcal{Q}$ is a family of location-scale transformations of probability measures in $\mathcal{Q}_0$. 
Let $w$ be some weight function representing the region of interest as $\{x : w(x)>0\}$, such that $P_\Q(w(X)>0)>0$ for all $\Q\in\mathcal{Q}_0$ and let $\Q^w$ be the conditional distribution of $X^w = X|w(X)>0$ if $X\sim\Q$. If $X$ has density $f(\cdot)$, then $X^w$ has density $f_w(\cdot) = f(\cdot)/P_\Q(w(X)>0)$. A location-scale transformation $(a + b X)^w$ of $X^w$ has density $f_w((\cdot-a)/b)/b$, which means that if $X^w \in \mathcal{Q}_0$, then $(a+bX)^w \in\mathcal{Q}$. In other words, the conditional distributions of location-scale transformations are also location-scale transformations. 
Let $\mathcal{Q}^w$ denote the set of all these conditional distributions corresponding to $\Q \in \mathcal{Q}$, and we assume for simplicity that the conditional distributions of the measures in $\mathcal{Q}_0$ are in $\mathcal{Q}_0$, so that $\mathcal{Q}^w \subset \mathcal{Q}$.

Because we are interested in the region $\{x : w(x)>0\}$, we define the concept of weight-scale invariance by considering the scaling properties of $S$ when restricted to $\mathcal{Q}^w$.

\begin{definition}\label{def:weight:scale:invariance}
Let $S$ be a proper scoring rule with respect to some class of probability measures $\mathcal{P}$ on $(\R,\mathcal{B}(\R))$, and assume that $\mathcal{Q}_0$ is a set of probability measures such that $\mathcal{Q}\subseteq \mathcal{P}\cap \mathcal{P}_S$. Suppose that $w$ is a weight function and let $\mathcal{Q}^w \subset \mathcal{Q}$ be the set of conditional probability measures as defined above. We say that $S$ is a locally weight-scale invariant scoring rule with respect to a $w$ on $\mathcal{Q}$ if $S$ is locally scale invariant on $\mathcal{Q}^w$. 
\end{definition}

\begin{remark}\label{remark:tail:scale}
Clearly, if $S$ is locally scale invariant, it is also locally weight-scale invariant. However, as we will see later, there are scoring rules that are locally weight-scale invariant but not locally scale invariant.
\end{remark}

Recall that the scoring rule $S$ is locally scale invariant if its scale function satisfies $s(\Q_{\thetab}) = \frac1{\sigma^2}s(Q_{(0,1)})$. Considering the restriction of $S$ to $\mathcal{Q}^w$, we have that $S$ is locally weight-scale invariant if $s(\Q_{\thetab}^w) = \sigma^{-2} s(\Q_{(0,1)}^w)$.

\begin{remark}
    If the interest is focused on the region $w(\cdot)>0$, we could simply use $S$ restricted to $\mathcal{Q}^w$ as a scoring rule, which means that we only consider the score based on the conditional distributions. However, this disregards the probability of $w(\cdot)>0$ which may be important in other situations. This probability is typically dependent on the scale, and local weight-scale invariance can thus be thought of as relaxation of local scale invariance where the scale dependence of the probability of $w(\cdot)>0$ is ignored. 
\end{remark}

In the special case where the weight function is the indicator weight function $w_u$ we refer to local weight-scale invariance as local tail-scale invariance. Specifically, we have the following definition.

\begin{definition}
Let $w_u$ be the indicator weight function in Eq.~\eqref{eq:qwf}. Under the same assumptions as in Definition~\ref{def:weight:scale:invariance}, we say that $S$ is a locally tail-scale invariant scoring rule on $\mathcal{Q}$ if for each $\thetab$ there exists ${u_{\thetab}\in\text{supp}(\Q_{(0,1)})\cap\text{supp}(\Q_{\thetab})}$ such that $S$ is locally scale invariant on $\mathcal{Q}^{w_u}$ for every $u\geq u_{\thetab}$. Similarly, we say that $S$ is locally lower tail-scale invariant on $\mathcal{Q}$ if $S$ is locally scale invariant on $\mathcal{Q}^{1-w_u}$ for every $u\leq u_{\thetab}$. 
\end{definition}

\subsection{Scaled weighted CRPS}

For any kernel score we can use the construction in Eq.~\eqref{eq:kernel:standard} for generalised kernel scores to construct a corresponding scaled version by using $h(x) = -\frac{1}{2}\log(x)$ as done for the SCRPS. By the following proposition, this can be done for the wCRPS. 
\begin{proposition}\label{prop:wcrps:kernelscore}
The wCRPS as defined in Eq. \eqref{eq:wCRPS} is a kernel score.
\end{proposition}
This proposition was shown previously by \citet{Allen2022EvaluatingScores}, and for completeness, a proof is given in Appendix~\ref{sec:appendix:proofs}.
Thus, by this proposition combined with the construction in Eq.~\eqref{eq:kernel:standard}, the scoring rule defined as
\begin{equation}
\begin{aligned}
    swCRPS(\p,y)&:=-\frac{E_\p[g_w(X,y)]}{E_{\p,\p}[g_w(X,X')]}-\frac{1}{2}\log(E_{\p,\p}[g_w(X,X')])\label{eq:score:swcrps}
\end{aligned}
\end{equation}
is proper, and can be a suitable alternative for evaluating extremes. 
As for the SCRPS, an alternative definition of the swCRPS would be to use Eq. \eqref{eq:scale-invariant} on the wCRPS. This definition was proposed in \citep{Vandeskog2022}, and since the wCRPS is a kernel score the two definitions  in fact coincide, up to an additive constant. 

No matter which of the arguments we use to define the swCRPS, it is not guaranteed that it is locally scale invariant. We now derive to what extent the swCRPS is scale invariant. 
The following assumption will be used on the distributions we consider. 
\begin{assumption}\label{assumption1}
The Borel probability measure $\Q$ on $\R$ has density $\exp(\Psi)$ with respect to Lebesgue measure for some twice differentiable function $\Psi$ such that the expectations $E_\Q[\Psi'(X)]$, $E_\Q[\Psi''(X)]$ and $E_\Q[(\Psi'(X))^2]$ are finite.
\end{assumption}

Under this assumption, we have the following result for the wCRPS and the swCRPS, which is proven in Appendix~\ref{sec:appendix:proofs}. 
\begin{proposition}\label{prop:wCRPS:scale:function}
Let $\mathcal{Q}$ be a set of location-scale transformed probability measures $\Q_{\thetab}$ satisfying Assumption \ref{assumption1}. Then the following holds:
\begin{enumerate}[(i)]
    \item The wCRPS with indicator weight function $w_u$ is neither locally scale invariant nor locally tail-scale invariant on $\mathcal{Q}$.
    \item The swCRPS as defined in Eq. \eqref{eq:score:swcrps} with indicator weight function $w_u$ is locally tail-scale invariant but not locally scale invariant on $\mathcal{Q}$.
\end{enumerate}
\end{proposition}

Other versions of the CRPS, the rCRPS and the rSCPRS as introduced in \cite{Bolin2022LocalRulesCustom}, are deinfed by inserting the kernel function
\begin{equation}
    g_c(x,y)=\begin{cases}
        |x-y|,&|x-y|<c\\
        0,&\text{otherwise}
    \end{cases}
\end{equation}
into the definition kernel scores in Eq. \eqref{eq:kernel:first} and into Eq. \eqref{eq:score:swcrps}. They showed that neither of these score were locally scale invariant. However, one can show their local weight-scale invariance.

\begin{proposition}\label{prop:rCRPS}
\begin{enumerate}[(i)]
    \item The rCRPS is not locally weight-scale invariant with respect to the indicator weight function $w^u(y)=1\{|y|<u\}$ for any $u>0$.
    \item The rSCRPS is locally weight-scale invariant with respect to $w^u$ for $u\leq \frac{c}{2}$.
\end{enumerate}
\end{proposition}

The steps taken above by using the construction in Eq.~\eqref{eq:kernel:standard} to create the swCRPS can be done on the more general threshold weighted kernel score, proposed by \citet{Allen2022EvaluatingScores}, and defined as
\begin{equation}\label{eq:Allen:twks}
twS_\rho(\p,y;v)=\frac{1}{2}E_{\p,\p}[\rho(v(X),v(X'))]-E_{\p}[\rho(v(X),v(y)]+\frac{1}{2}\rho(v(y),v(y))
\end{equation}
where $\rho$ is a continuous negative-definite kernel and $v(x)-v(x')=\int_x^{x'} w(t)dt$ is a measurable chaining function. This results in a scaled score 
\begin{equation}\label{eq:Allen:twks:scaled}
stwS_{\rho}(\p,y;v)=-\frac{E_{\p}[\rho(v(X),v(y)]}{E_{\p,\p}[\rho(v(X),v(X'))]}-\frac{1}{2}\log(E_{\p,\p}[\rho(v(X),v(X'))]).
\end{equation}
However, we focus on the more specific swCRPS to keep the exposition reasonably short and leave further investigations of these scaled scores for future research.

As mentioned earlier, one main benefit of kernel scores is that the expected values can be computed with Monte Carlo methods without knowing the true distribution. However, under an assumed model distribution, the kernel scores can often be computed analytically. This is in particular the case for the Generalised Extreme Value (GEV) distribution, as shown in Appendix~\ref{sec:models}.

\subsection{Local weight-scale invariance of the censored likelihood score}

It follows from the definition that a sum of two proper scores is also proper, more generally that a weighted sum of several proper scores is proper, and  additionally that if at least one of the summands is strictly proper, then the sum is also strictly proper. This can be used to propose alternative scores which may be less likely to suffer the problems of overweighting heavy tails as discussed above for the quantile scores \citep[see also, e.g.,][]{Lerch2017}. 

If one is interested in good prediction of values exceeding some threshold $u$, there are two complementary questions: a) how well does a method predict the sizes of the excesses which occur, and  b) does the method give good prediction of the probability of exceedance. Here a) can be approached by only scoring the values of the excesses which actually occur, using the conditional distribution of these, $f(x)/(1-F(u))$. Question b) instead is a binary prediction problem and one can use any of the many scoring rules for such problems. Further, it seems reasonable to use related scoring rules for a) and b).

\citet{Diks2011} uses this idea to combine the logarithmic score for excesses
\begin{equation*}
S(\p, y) = 1\{y > u\}\log (f(y)/(1-F(u)),    
\end{equation*}
where $f$ and $F$ are the density and cumulative distribution functions corresponding to $\p$, 
with the binary logaritmic scoring rule for exceedance of the threshold $u$,
\begin{equation*}
    S(\p,y) = 1\{y \leq u\}\log(F(u)) +  1\{y > u\}\log (1-F(u)),
\end{equation*}
to obtain the so-called censored likelihood score
\begin{equation}\label{eq:lsq}
    LS_{u}(\p,y) = 1\{y\leq u\}\log(F(u))+1\{y > u\}\log(f(y)).
\end{equation}
This is a proper scoring rule and is studied in simulations and application in the next two sections. We also have the following result, proven in Appendix~\ref{sec:appendix:proofs}.

\begin{proposition}\label{prop:lsq} 
Let $\mathcal{Q}$ be a set of location-scale transformed probability measures $\Q_{\thetab}$ satisfying Assumption \ref{assumption1}. Then the censored likelihood score in Eq. \eqref{eq:lsq} is locally tail-scale invariant on $\mathcal{Q}$, but in general not locally scale invariant. 
\end{proposition}

The above type of weighted scoring rule was generalised in \citet{HolzmanKlar2017} as follows. First, let $p_w$ represent the re-normalised density of $p$ with respect to weight function $w$, i.e.
\begin{equation}
    p_w(y)=\frac{w(y)p(y)}{\int_\R w(t)p(t)dt}.
\end{equation}
If $S_0$ is a proper scoring rule, then
\begin{equation}\label{eq:score:outcome-weighted}
    owS(\p,y;w)=w(y)S_0(\p_w,y),
\end{equation}
where $\p_w$ is the distribution with density $p_w$, is a proper scoring rule. This is called an outcome weighted scoring rule in \citet{Allen2022EvaluatingScores}.

Next, for a strictly proper scoring rule ${\bf s}(\alpha,z)$ for the success probability $\alpha\in(0,1)$ of a binary outcome variable $z\in\{0,1\}$, the score
\begin{equation}\label{eq:score:success-prob}
    S_{\bf s}(\p,y;w)=w(y){\bf s}\left(\int_\R p(t)w(t)dt,1\right)+(1-w(y)){\bf s}\left(\int_\R p(t)w(t)dt,0\right)
\end{equation}
is said to be a localising proper weighted scoring rule for the density forecast $p$. Finally, adding Eq. \eqref{eq:score:outcome-weighted} and Eq. \eqref{eq:score:success-prob} yields a proper scoring rule
\begin{equation}\label{eq:score:combined}
    \hat{S}(\p,y,w) = S_{\bf s}(\p,y;w)+owS(\p,y;w).
\end{equation}
 With $S_0$ and ${\bf s}$ as log scores, the score $\hat{S}$ as defined in Eq. \eqref{eq:score:combined} becomes a censored likelihood with a general weight function $w$ instead of the weight function $w_u$ defined in Eq.~\eqref{eq:qwf}.

\begin{proposition}\label{prop:general_cencoredlikelohood}
Given that $\lambda({x : w(x) = 1}) > 0$,  where $\lambda$ is the Lebesgue measure, the general censored Likelihood score
    \[CLog(\p,y;w)=w(y)\log\left(\int_\R p(t)w(t)dt\right)+(1-w(y))\log\left(1-\int_\R p(t)w(t)dt\right)+w(y)\log(p_w(y))\]
    is locally weight-scale invariant with respect to $\tilde{w}(x)=1\{w(x)=1\}$.
\end{proposition}

\begin{proposition}\label{prop:combined-scores} 
Suppose that $S_0$ and $S_1$ are locally weight-scale invariant with respect to $w$. Then
\begin{enumerate}[(i)]
    \item If $w$ is an indicator weight function, then $owS$ as defined in Eq. \eqref{eq:score:outcome-weighted} is locally weight-scale invariant with respect to $w$.
    \item $S_0+S_1$ is locally weight-scale invariant with respect to $w$.
\end{enumerate}
\end{proposition}

\section{Simulation studies}\label{sec:results:simulation}
 
This section contains two examples with simulated data, which illustrate the effect of scale dependence for extremes using the wCRPS and the swCRPS.

\subsection{Benchmark example}\label{sec:res:sim:benchmark}

We start by comparing the predictive ability of CRPS against SCRPS in terms of tail regimes. A comparison using CRPS was previously done using a benchmark example in \citet{TAILLARDAT2022}. They considered the hierarchical model
\begin{equation*}
    \begin{cases}
    Z \overset{d}{=} \text{Gamma}(\xi^{-1},\xi^{-1})&\\
    Y|Z \overset{d}{=} \text{Exp}(Z) 
    \end{cases}
\end{equation*}
where $1>\xi>0$, $\text{Gamma}(\alpha,\beta)$ denotes a gamma distribution with shape parameter $\alpha$ and rate parameter $\beta$, and $\text{Exp}(\delta)$ denotes an exponential distribution with scale parameter $\delta$. Using this as the true model, four forecasting models are compared. The ideal model is an exponential model using observed values of $Z$. The extremist model instead underestimates the scale parameter of the exponential distribution, assuming it to be $Z/\nu$, where $\nu>1$. The climatological model has no information about $Z$ but instead uses the unconditional and distributionally equivalent Generalized Pareto distribution, $\text{GP}(1,\xi)$, with scale parameter one and shape parameter $\xi$. And finally, the $\tau$-informed model is a mixture distribution between the ideal and climatological model distributions. The models are listed in Table \ref{tab:benchmark:models}.

\begin{table}[t]
    \centering
    \caption{Forecasting models for benchmarking, taken from \citet{TAILLARDAT2022}. The true model is $Y|Z \overset{d}{=} \text{Exp}(Z)$, where $Z \overset{d}{=} \text{Gamma}(\xi^{-1},\xi^{-1})$, $1>\xi>0$.}
    \label{tab:benchmark:models}
    \begin{tabular}{ll}
    \toprule
      Forecast & Density \\
     \hline
      Ideal, $\p_{ideal}$   &  $f_{\text{Exp}(Z)}$\\
      Climatological, $\p_{clim}$ & $f_{\text{GP}(1,\xi)}$\\
      $\tau$-Informed, $\p_\tau$& $\tau f_{\text{Exp}(Z)}+(1-\tau)f_{\text{GP}(1,\xi)}$\\
      Extremist, $\p_{extr}$ & $f_{\text{Exp}(Z/\nu)},\;\nu>1$\\
      \bottomrule
    \end{tabular}
\end{table}

The CRPS for the extreme forecast and the $\tau$-informed forecast are
\begin{equation}\label{crps:extr}
    CRPS(\p_{extr},y)=-y-\frac{2\nu}{\delta}\exp\left(-\frac{\delta y}{\nu}\right)+\frac{3\nu}{2\delta}
\end{equation}
and
\begin{equation}\label{crps:lambda}
\begin{aligned}
   CRPS(\p_\tau,y)=&-y-\frac{\tau^2}{2\delta}-\frac{2\tau}{\delta}[\exp(-\delta y)-1]+\frac{2(1-\tau)}{1-\xi}\left[1-(1+\xi y)^{\frac{\xi-1}{\xi}}\right]\\
    &-\frac{(1-\tau)^2}{2-\xi}-\frac{2\tau(1-\tau)\exp\left(\frac{\delta}{\xi}\right)}{\xi^{\frac{1}{\xi}}\delta^{\frac{\xi-1}{\xi}}}\Gamma_u\left(\frac{\xi-1}{\xi},\frac{\delta}{\xi}\right),
\end{aligned}
\end{equation}
respectively, as shown by \cite{TAILLARDAT2022}. Here, $\Gamma_u(a,\tau)=\int_\tau^\infty t^{a-1}e^{-t}dt$ denotes the upper incomplete gamma function. From the CRPS scores of Eq. \eqref{crps:extr} and \eqref{crps:lambda}, one can compute the SCRPS by first noting that $E_{\p_\text{extr},\p_\text{extr}}[|X-X'|]=\frac{\nu}{\delta}$,
\begin{align*}
    E_{\p_\tau,\p_\tau}[|X-X'|]&=2\left(\frac{\tau}{\delta}-\frac{\tau^2}{2\delta}+\frac{1-\tau}{1-\xi}-\frac{(1-\tau)^2}{2-\xi}-\frac{2\tau(1-\tau)\exp\left(\frac{\delta}{\xi}\right)}{\xi^{\frac{1}{\xi}}\delta^{\frac{\xi-1}{\xi}}}\Gamma_u\left(\frac{\xi-1}{\xi},\frac{\delta}{\xi}\right)\right),
\end{align*}
and
    $E_{\p_\text{clim},\p_\text{clim}}[|X-X'|]=\frac{2}{(2-\xi)(1-\xi)}$,
and then using Eq. \eqref{eq:scrps} to compute the SCRPS.

We simulate $10^6$ observations of the model, for two choices of $\xi$, $\xi=0.25$ and $\xi=0.5$. The higher $\xi$ is, the lower is the variability in $Z$, and thus the scale of $Y|Z$. Comparing the ratio of the mean scores of the different models to the corresponding means for the ideal model for CRPS and SCRPS creates an ordering of the models, as seen in Table \ref{tab:res:benchmark}. Both scores order as expected within the $\tau$-informed models and within the extremist models, with decreasing $\tau$ and increasing $\nu$ respectively. However, the ordering within the two models depends on the choice of scoring rule. Hence, by choosing the scoring rule, one makes an implicit choice regarding which aspects of the model that are important. For example, for $\xi=0.25$, the climatological forecast is worse than the extremist forecast with $\nu = 1.8$ according to the CRPS but better according to the SCRPS, probably due to how the CRPS penalises the forecasts with larger uncertainty more than the SCPRS does, while the SCRPS prefers the models with the same relative error over the climatological model.

The sensitivity and power of the respective scores are also of interest, for different values of $\nu$. For a score $S$, the extreme model is compared against the ideal model on a simulated dataset with 1000 observations using a two-sided pairwise t-test. This is repeated 1000 times with different simulated datasets. The SCRPS better identifies the ideal model from the extremist model, as seen in Figure
\ref{fig:wilcox:benchmark}, with higher power when comparing the ideal and the extremist model. Moreover, the power of the SCRPS does not change for increased values of $\xi$ while the power of the CRPS decreases.

\begin{table}[t]
    \centering
    \caption{Ratio of the mean CPRS and mean SCRPS with respect to the corresponding means for the ideal forecast for two different choices of $\xi$. The mean was taken over $10^6$ independently simulated observations.}
    \label{tab:res:benchmark}
    \begin{tabular}{lcccc}
    \toprule
         \multirow{2}{*}{Forecast} & \multicolumn{2}{c}{$\xi=0.25$} & \multicolumn{2}{c}{$\xi=0.50$} \\
         \cmidrule(r){2-3}\cmidrule(r){4-5}
         &CRPS&SCRPS&CRPS&SCRPS\\
         \hline
      Ideal, $\p_{ideal}$   & $100\%$   & $100\%$& $100\%$   & $100\%$ \\
      Extremist, $\nu=1.1$ & $100.48\%$   & $100.41\%$ & $100.47  \%$   & $100.39\%$\\
      $0.75$-Informed & $100.89\%$ & $101.28\%$ & $102.14  \%$   & $104.26  \%$\\
      $0.5$-Informed & $103.56\%$ & $103.76\%$ & $108.47  \%$   & $109.93  \%$\\
      Extremist, $\nu=1.4$ & $106.67\%$   & $104.62\%$ & $106.64 \%$   & $104.35\%$\\
      $0.25$-Informed & $108.02\%$ & $107.20\%$ & $119.00\%$   & $116.31\%$\\
      Climatological, $GP(1,\xi)$ & $114.27\%$ & $113.67 \%$ & $133.72\%$   & $131.58  \%$\\
      Extremist, $\nu=1.8$ & $122.87\%$   & $112.69\%$ & $122.83\%$   & $111.94\%$\\
      \bottomrule
    \end{tabular}
\end{table}

\begin{figure}[t]
    \centering
    \includegraphics[width=0.8\textwidth]{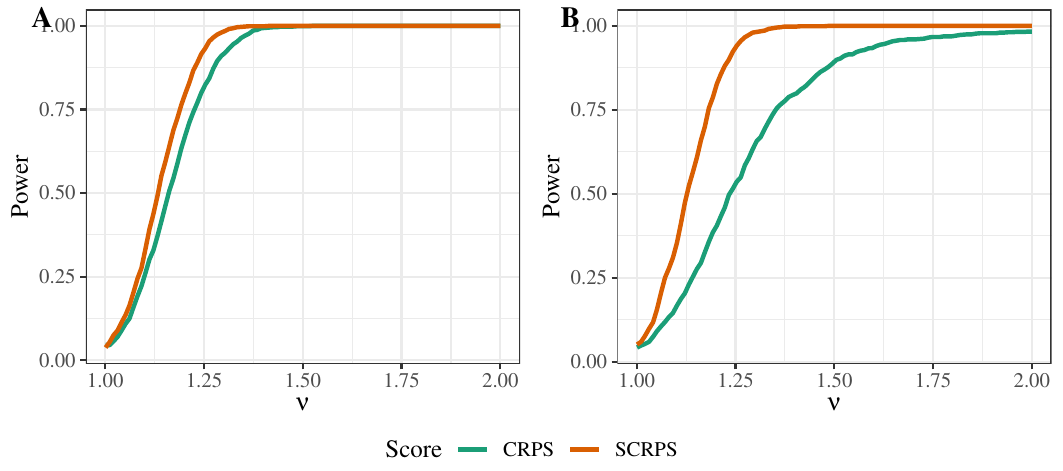}
    \caption{Benchmark model simulation showing power of a two-sided pairwise t-test when scoring 1000 independently simulated stationary time series of length 1000 from the ideal model $Y|Z\sim \text{Exp}(Z)$, $Z\sim\text{Gamma}(\xi^{-1},\xi^{-1})$ using the extremist model with parameter $\nu\in(1.001,2)$ and shape parameters $\xi=0.25$ (left) and $\xi=0.5$ (right).}   
    \label{fig:wilcox:benchmark}
\end{figure}

\subsection{Score dependence on scale and threshold}\label{sec:res:sim:escore:scale}

\begin{figure}[t]
     \centering
    \includegraphics[width=0.8\textwidth]{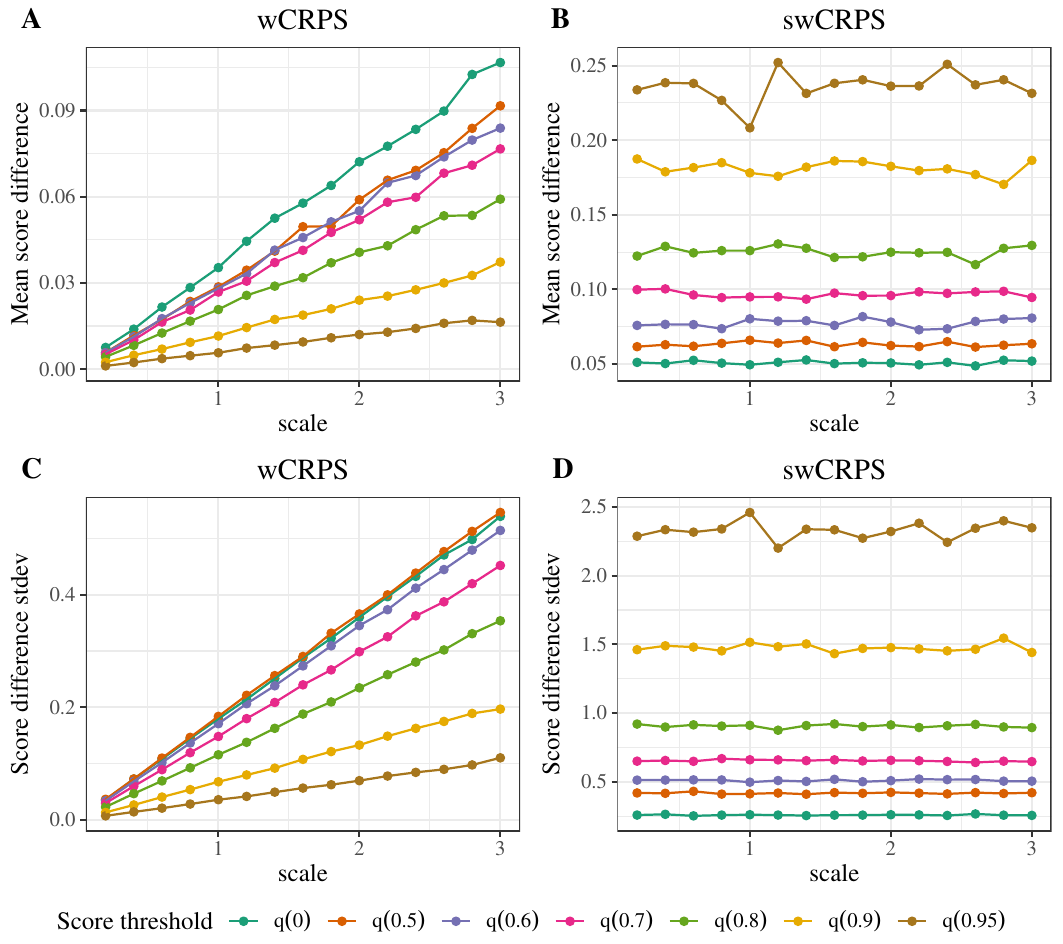}
        \caption{Simulated mean (top) and standard deviation (bottom) of the score difference ${S(\Q_{\thetab},\Q_{\thetab})-S(\p,\Q_{\thetab})}$ using wCRPS (left) and swCRPS (right) for two predictions of a random variable $X\sim \Q_{\thetab}=GEV(\mu=0,\sigma,\gamma=0.12)$, with $\p=GEV(\mu=0,2\sigma,\gamma=0.12)$, as functions of the scale parameter $\sigma$ for different thresholds $q(p)$, chosen as the $p$-th quantile from $\Q$. Threshold $q(-\infty)$ results in the unweighted scores, CRPS and SCRPS.}
        \label{fig:res:sim:escore:scale}
\end{figure}
For an observation from $X\sim \Q_{\thetab}=GEV(\mu=0,\sigma,\gamma=0.12)$, let us consider the effect of choosing a wrong model by computing the expected score difference $S(\Q_{\thetab},\Q_{\thetab})-S(\p,\Q_{\thetab})$ where ${\p=GEV(\mu=0,2\sigma,\gamma=0.12)}$. Here the value $\gamma=0.12$ of the shape parameter is chosen to mimic the behaviour of the extreme rainfall events which are studied in Section~\ref{sec:results:casestudy}. 
The mean and standard deviation of the score difference is computed from 50.000 simulated observations of $X$. For the weighted scores, the thresholds $q(p)$ are obtained as the $p$-th quantile of $\Q_{\thetab}$,
$
    q(p)=\mu+\frac{\sigma}{\gamma}\left((-\log p)^{-\gamma}-1\right)
$
for $p>0$, and $q(0)$ represent the unweighted scores.
Figure \ref{fig:res:sim:escore:scale} shows that for the wCRPS, differences increase with increased $\sigma$ while they remain stable for the swCRPS. The behaviour is similar for each choice of threshold. Variability increases with scale for the wCRPS but it remains constant for the swCRPS. Thus the local tail-scale invariance does not only affect the mean score but also the score variance. The fact that the mean and variance of scores remain the same for different scales simplifies understanding and use of empirical distributions of scores from predictions with different scales. It can also help scoring scaled locations simultaneously, as we will see in the following section.

\subsection{Scaling effect on expected scores}\label{sec:res:sim:escore}

\begin{figure}[t]
     \centering
          \includegraphics[width=0.9\textwidth]{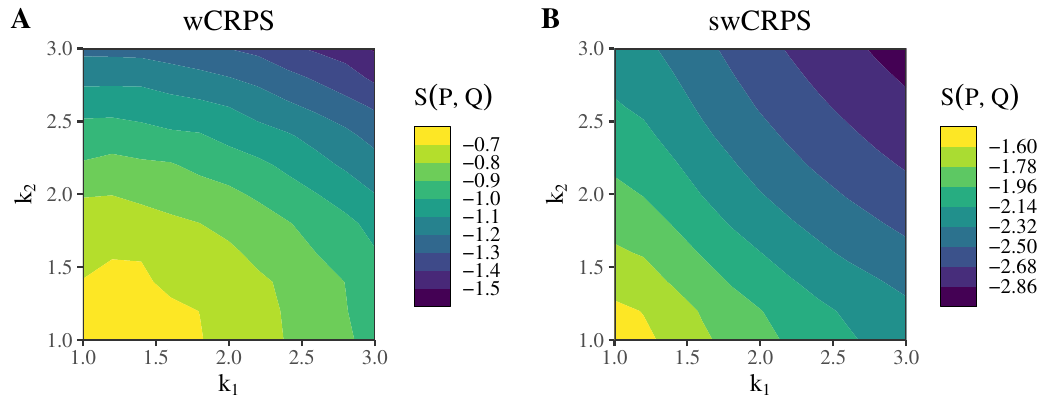}
        \caption{Simulated expected score $S(\p,\Q)$ using wCRPS and swCRPS for a pair $(X_1,X_2)$ of random variables with $X_i\sim GEV(\mu_i,\sigma_i,\gamma)$, as functions of $k_i$, $i=1,2$ using a prediction that has the correct location parameters, $\mu_i$, and scale parameters $\widehat{\sigma_i}=k_i\sigma_i$, $i=1,2$. For the true model, $\mu_1=\mu_2=0$, $\gamma=0.12$, $\sigma_1 = 1.5$ and $\sigma_2 = 3$. The weight function was chosen as the 0.90 quantile for each score.}.
        \label{fig:res:sim:escore}
\end{figure}

As mentioned above, one often evaluates forecasts at multiple locations though average scores. If the individual scores are proper, their average is also a proper score. However, if the distributions at these locations have different scales, the use of a scale dependent scoring rule leads to an implicit ranking of the importance of the different locations, in the sense that some locations might be more important to predict well to get a high average score.  

As an example, consider two observations following GEV distributions that differ only in scale, $X_i\sim \Q_{\thetab_i}=GEV(0,\sigma_i,0.12)$, $i=1,2$, with $\sigma_1=1.5$ and $\sigma_2=3$. Using the predictions $\p_i=GEV(0,\widehat{\sigma_i},0.12)$, where $\widehat{\sigma_i}$ is an estimated scale parameter, the paired observations are scored using the scoring rule $S(\p,y)=\frac{1}{2}\left(S_1(\p_1,y_1)+S_2(\p_2,y_2)\right)$, where $S_1$ and $S_2$ are either wCRPS scores,  or else swCRPS scores with  indicator weight functions $w_{q_1}$ and $w_{q_2}$ for the two observations. 
Here $q_1$ and $q_2$ are chosen as the 90th quantiles of $\Q_{\thetab_1}$ and $\Q_{\thetab_2}$ respectively. The expected score was computed as the mean of 100.000 simulations.

Figure \ref{fig:res:sim:escore} shows that the expected score $S(\p,\Q)$ with $\widehat{\sigma_i}=k_i\sigma_i$ for wCRPS is more sensitive to changes in $k_2$ than in $k_1$, meaning that it is more important to have a good prediction for $X_2$. 
For the swCRPS, the scores are on the other hand quite symmetrical in $k_1,k_2$, so the predictions $\p_1,\p_2$ are scored more equally. Figure \ref{fig:res:sim:escore} further shows that the swCRPS is close to being locally scale invariant while the wCRPS is not. Finally, note that the symmetry in $k_1,k_2$ for all considered values indicates that the score difference does not depend on the scale even for large model misspecifications, which is not guaranteed to hold by the definition of local scale invariance. This is not uncommon, and for example also holds for Gaussian distributions as can be seen in \citet[][Figure 1]{Bolin2022LocalRulesCustom}.

\section{Case studies}\label{sec:results:casestudy}

This section uses data on water levels, precipitation, and air pollution to compare scoring rules in practice. The models used were fitted by maximising the log likelihood in the first two case studies and through the INLA approach~\citep{inla2009} in the final study.

\subsection{Extreme water levels}

In this section, we consider data containing annual maximal water levels at five representative stations in the Great Lakes; Lake St. Clair, Lake Michigan-Huron, Lake Ontario, Lake Superior, and Lake Erie. The dataset is provided by NOAA and dates from 1918 to 2020. With different sizes and depths of the lakes, the data differs in scale, making it interesting to see the effect of misspecifications when using wCRPS and swCRPS.

First, a stationary GEV model was fitted to the five representative stations, assuming that the behaviour of the lake-wide average water levels has not changed during the observed time series. The estimated parameters are listed in Table~\ref{tab:great:lakes:gev:fit}. The shape parameter at all locations is negative, suggesting that the annual maxima follow a Weibull distribution. The density functions of the estimated GEV distributions are shown in Figure \ref{fig:pdf:gev:lakes} in the appendix. Next, the $PGEV_\lambda$ model, with a trend in $\lambda$ and temperature as covariate, was fitted to each of the stations. For further information on the PGEV model, see  \citet{helgak}. The temperature used was a lowess smoothing of the yearly average Northern Hemispheric temperature, obtained from NOAA~\citep{NOAANationalCentersforEnvironmentalInformation2019ClimateSeries}. The fit suggested that Lake Michigan-Huron and Lake Superior did not have a trend in the expected number of extreme water levels, $\lambda$, but trends existed in the other three lakes, Lake St. Clair, Lake Ontario and Lake Erie, using significance level $\alpha=0.05$. From the Lake System profile, the stationarity might be explained by flow from both Lake Michigan-Huron and Lake Superior to the other lakes. The number of dams in the great lakes exceeds 7.000 so the effect of them on the data is hard to visualise.

\begin{table}[t]
    \centering
    \caption{Parameter estimates (standard deviations) using stationary GEV distribution.}\label{tab:great:lakes:gev:fit}
    \begin{tabular}{rccccc}
    \toprule
 Lake & $\mu$ & $\sigma$ &$\gamma$ & station id\\\hline
St. Clair &175.108 (0.038) &  0.349 (0.027)& -0.285 (0.065) & 1\\
Michigan-Huron &176.469 (0.044)& 0.395 (0.033)& -0.283 (0.082)& 2\\
Ontario & 74.990 (0.034)& 0.322 (0.024)& -0.285 (0.053)& 3\\
Superior &183.524 (0.019)&  0.175 (0.014)& -0.404 (0.063)& 4\\
Erie&174.280 (0.038)&0.355 (0.027)& -0.348 (0.060)& 5\\
\bottomrule
    \end{tabular}
\end{table}

In \citep{glerlnoaa} it is noted that ``since September 2014, all of the Great Lakes have been above their monthly average levels for the first time since the late 1990s''. However, they come to the same conclusion, that probably the Lake Superior and Lake Michigan-Huron will remain stationary around the mean levels while there might be non-stationarity in the remaining three lakes. 

The above parameter estimates are used through simulation to compare the effect of parameter misspecification on the mean and standard deviation of the score differences. Further, since the stationary stations, Lake Superior and Lake Michigan-Huron are also the ones that differ most in scale, these are used in the simulation study to compare how often one misspecified model is preferred over another equally misspecified model using different scores.

\subsubsection{Simulation}

Using the estimated parameters from Table \ref{tab:great:lakes:gev:fit}, we simulate observations from the five different stations. For each station, 1000 independent stationary time series of length 100 are simulated. These time series have been scored for two models. The first model uses the true parameters $\mu,\sigma,\gamma$ from the simulation and the second model uses the true location and shape parameters $\mu$ and $\gamma$, but perturbs the scale parameter by a factor $k=1.5$. The scale parameter of the second model therefore becomes $1.5\sigma$. We see that for the parameters from station $4$, Lake Superior, both the mean score difference and its variation is smaller than for the other stations when using wCRPS (Figure~\ref{fig:mean:scores:lakes}). This difference between stations disappears when instead using swCRPS on the same data. 

For a given perturbation factor $k$, the score difference at station $i$ is defined as
$\Delta_i:=S_{\sigma_i}-S_{k\sigma_i}$
where $S_{\hat{\sigma}}$ is the score using estimated scale parameter $\hat{\sigma}$, and $\sigma_i$ is the true scale parameter at station $i$ (i.e., the parameter that was used when simulating the data). The factor $k$ describes the estimated scale parameter's error relative to the true scale parameter. For scale invariant scores, the expected difference will only depend on the perturbation factor and not on the scale itself.

\begin{figure}[t]
    \centering
    \includegraphics[width=0.8\textwidth]{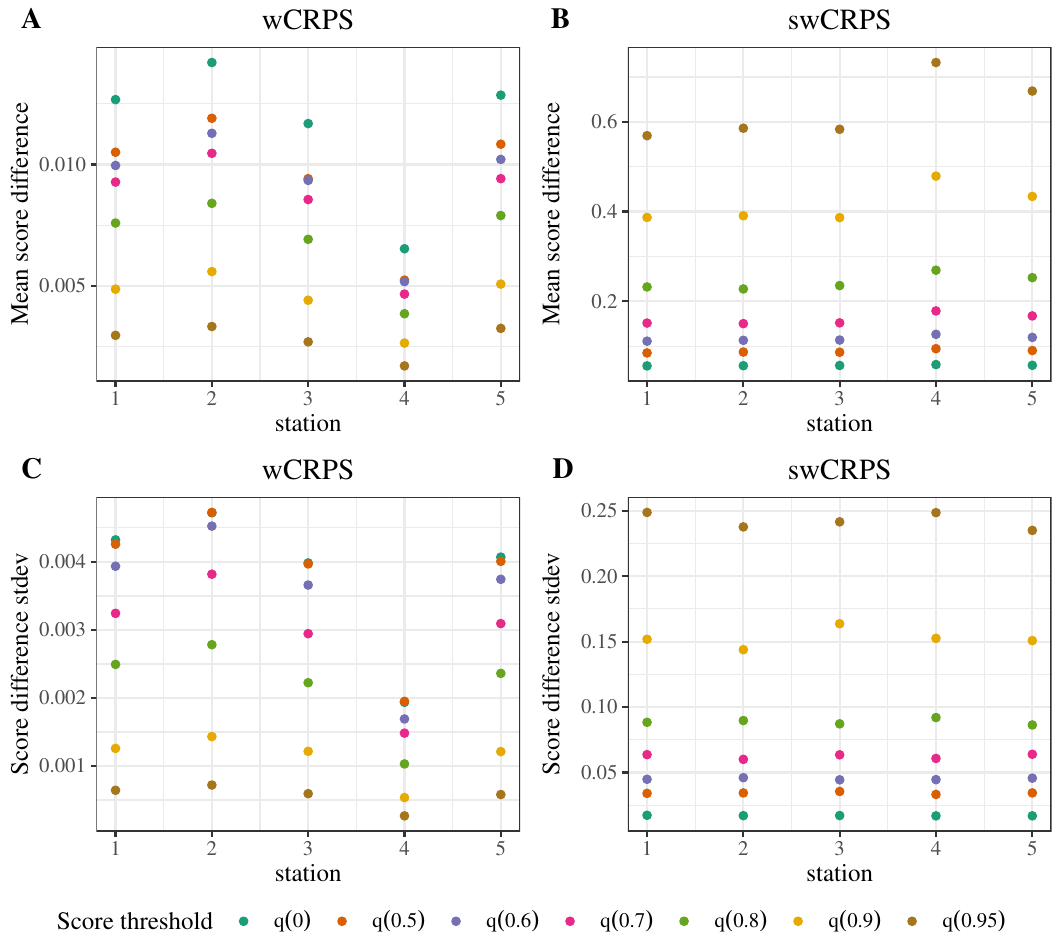}
    \caption{Mean and standard deviations of score differences, $\Delta_i$, at stations $i\in\{1,2,3,4,5\}$ with $k=1.5$, using different types of CRPS scores on simulated data using the estimated parameters from the stations listed in Table \ref{tab:great:lakes:gev:fit}}
    \label{fig:mean:scores:lakes}
\end{figure}

Assume we want to compare two model predictions, each for two locations: Lake Michigan-Huron with $\sigma_2=0.395$ and Lake Superior with $\sigma_4=0.175$. Denote the models A and B, where
\begin{equation*}
    \begin{aligned}
\text{Model A:} &\; \hat{\sigma}_2=\sigma_2,\;\hat{\sigma}_4=k\sigma_4\\
\text{Model B:} &\; \tilde{\sigma}_2=k\sigma_2,\;\tilde{\sigma}_4=\sigma_4        
    \end{aligned}
\end{equation*}
and $k>0$ is a perturbation factor.

We want to compare score A against score B, i.e.,
$S_A:=S_{\hat{\sigma}_2}+S_{\hat{\sigma}_4}$
against
$S_B:=S_{\tilde{\sigma}_2}+S_{\tilde{\sigma}_4}$.
However, $S_A-S_B=\Delta_2-\Delta_4$
so it suffices to compare the differences $\Delta_i$ at the stations. For this, 1000 time series of length 100 were simulated using the estimated parameters of Lake Superior and Lake Michigan-Huron. The perturbation factor was fixed at $k=1.5$. Figure \ref{fig:sup:michhur} shows how often model A was preferred over model B using different scores. Since the models both have one correct scale parameter, and the same proportional error in the other parameter, the models A and B should be chosen with equal probability if not influenced by scale. This corresponds to choosing A over B $50\%$ of the time. Using the wCRPS gives a high bias toward model $A$ while the swCRPS yields a more fair comparison of models $A$ and $B$, as expected. The deviation from the 0.5 line for large thresholds for the swCRPS is due to the difference in shape parameters at the simulated stations. If we assume that the shape parameters are the same for the two stations, this deviation disappears.

\begin{figure}[t]
    \centering
    \includegraphics[width=0.8\textwidth]{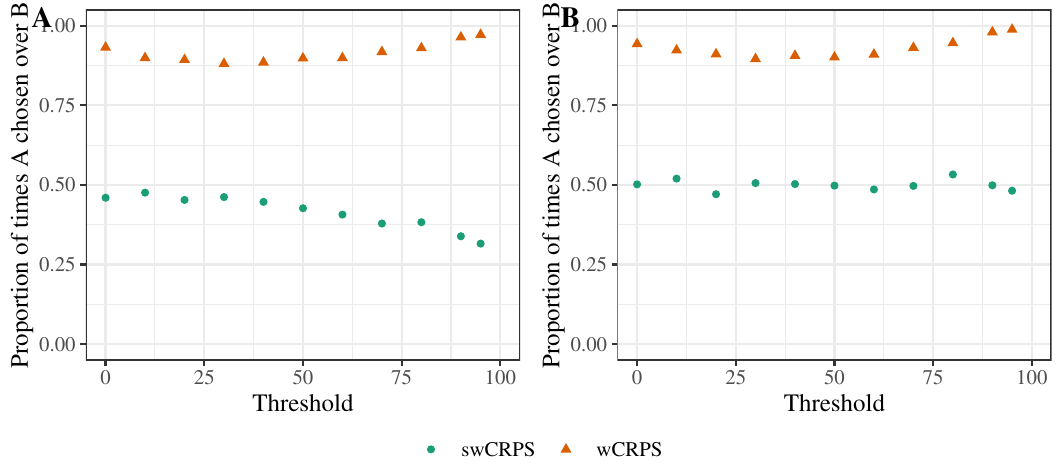}
    \caption{Proportion of times model A was preferred over model B when simulating time series of length 100 when using wCRPS and swCRPS, with individual shape parameter (left) and joint shape parameter (right). The score threshold is chosen as the $p\%$ quantile.}
    \label{fig:sup:michhur}
\end{figure}

\subsection{Extreme rainfall events and climate change}

\begin{figure}
     \centering
     \includegraphics[width=0.8\textwidth]{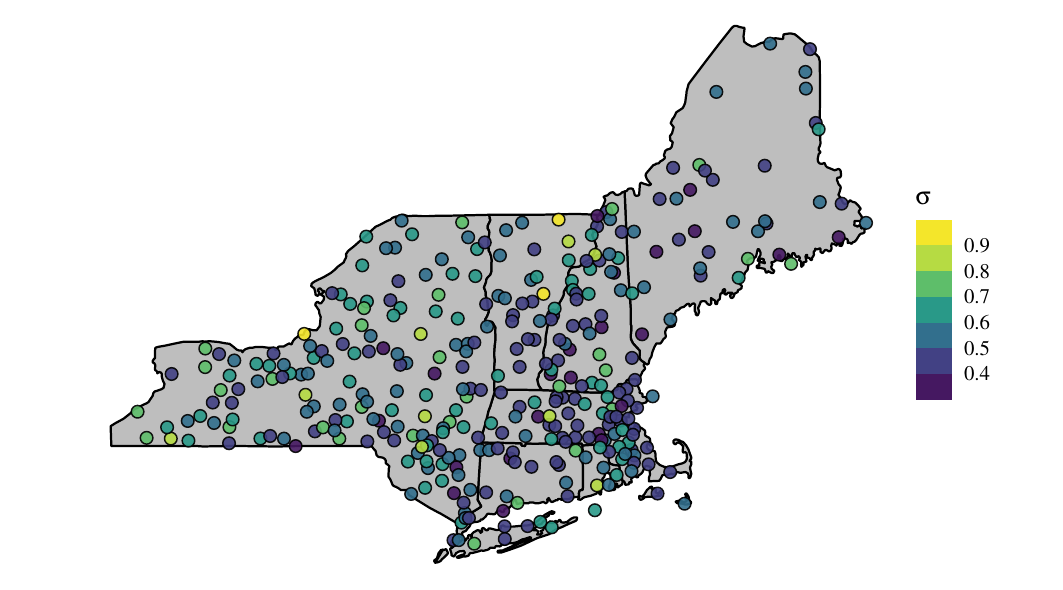}
     \caption{Estimated scale parameter from a $PGEV$ model with trend in the frequency parameter $\lambda$ for stations in the Northeastern U.S.}\label{fig:map:sigma}
\end{figure}

In this section we use the  scoring rules introduced above to compare five models for extreme rainfall in the Northeastern United States. The dataset is a part of NOAA Atlas 14 Volume 10, part 3 \citep{noaatsd} and contains the annual maximum rain from 685 stations in the Northeast with time series ranging from approximately 1900 to 2014. Only stations with at least 60 years of data are used. We call this data set the NE-data. The estimated scale parameters differ between stations (Figure \ref{fig:map:sigma}), meaning that the stations with higher scale parameters will be given greater weight if a scale dependent scoring rule such as CRPS and wCRPS is used.

\subsubsection{Model and inference details}

The four models for annual maxima rainfall compared are (i) Gumbel: a GEV distribution with shape parameter zero, (ii) GEV: a GEV distribution with shape parameter different from zero, (iii) GEV$_\mu$: a GEV distribution with trend in the location parameter, i.e. $\mu(t)=\mu_0+\mu_1 t$, and (iv) PGEV$_\lambda$: a PGEV distribution with trend in the frequency parameter, i.e. $\ln\lambda(t)=\lambda_0+\lambda_1 t$. For both GEV$_\mu$ and PGEV$_\lambda$, the covariate $t$ is a lowess smoothing of the yearly average Northern Hemispheric temperature, obtained from NOAA~\citep{NOAANationalCentersforEnvironmentalInformation2019ClimateSeries}.

The parameters of the GEV and the PGEV models are estimated by maximising the log-likelihood. 
Both models assume a single regional shape parameter and station-wise location and scale parameters. The estimated scale parameter varies over space, motivating the use of locally (tail-) scale invariant scores. The thresholds for the weighted scores are determined empirically from the time series as the $p$-th quantile.

Besides overall mean scores, we also consider the distribution of scores. For station $i$, with $N_i$ observations ${\bm y}=(y_1,...,y_{N_{i}})$, the mean score is
\begin{equation*}
    S_i(\p,{\bm y})=\frac{1}{N_i}\sum_{j=1}^{N_i}S(\p_{ij},y_j).
\end{equation*}
where $\p_{ij}$ is a predictive distribution for observation $j$ at station $i$. The station-wise score difference between two predictions $\p$ and $\Q$, at station $i$, is denoted
${\Delta_i(\p,\Q)=S_i(\p,{\bm y})-S_i(\Q,{\bm y})}$.
A positive difference means that prediction $\p$ scored better, and a negative difference means that prediction $\Q$ scored better. 
The standard way to evaluate if $\p$ is a better model than $\Q$ for time series is to perform a Diebold-Mariano test~\citep{Diebold1995ComparingAccuracyb} on difference of this kind, that takes into account correlation in the data. However, for the station-wise differences, the spatial correlation is low, and we therefore  perform pairwise t-test of equal predictive performance to compare two models. 

The variances of the station-wise average scores are affected by the varying scale parameters in the data, but also by the lengths of the time series, which are slightly different for the different stations. However, since only time series of length greater than 60 years are used, the latter difference is small.

\subsubsection{Results}

\begin{table}[tb]
    \centering
    \caption{Mean scores for four different models for NE-data, with means obtained as the mean of the mean scores at the individual stations. The scores where t-test showed significance between the PGEV$_\lambda$ and GEV$_\mu$ are shown in bold.}
    \label{tab:neusa:scores}
    \resizebox{\linewidth}{!}{
    \begin{tabular}{cccccccccc}
    \toprule
         &  \multirow{2}{*}{LS} & \multirow{2}{*}{CRPS} & \multirow{2}{*}{SCRPS} & \multicolumn{2}{c}{LS$_q$} & \multicolumn{2}{c}{wCRPS} & \multicolumn{2}{c}{swCRPS}\\
         \cmidrule(r){5-6}\cmidrule(r){7-8}\cmidrule(r){9-10}
         &&&&$90\%$&$99\%$&$90\%$&$99\%$&$90\%$&$99\%$\\
         \hline
      Gumbel & -1.148&-0.483 &-1.048 & -0.427 &-0.0903 & -0.0865 &  -0.010178 & -0.1556&-1.531 \\
      GEV & -1.129&-0.481 &-0.966 &-0.419& -0.0856 &-0.0860& -0.010201 & -0.0795&0.893 \\
      GEV$_\mu$ &-1.108&-0.473&-0.958& -0.415&-0.0850 & -0.0857&-0.010199 & -0.0644&0.906 \\
      PGEV$_\lambda$&{\bf -1.107} & {\bf -0.472} & {\bf -0.956} &  {\bf -0.414}& {\bf -0.0847}& {\bf -0.0854}& -0.010195& {\bf -0.0603}& 0.901\\
  \bottomrule
    \end{tabular}}
\end{table}

As shown in Table~\ref{tab:neusa:scores}, 
the PGEV$_\lambda$ model gives the highest overall mean scores, while the Gumbel model has the lowest score, for all statistically significant scores.
Figure \ref{fig:res:ne:ttest} shows  p-values from two-sided t-tests, where one can note that all scoring rules reject the null hypothesis except for very high thresholds, where only the LS$_q$ score rejects.

\begin{figure}[t]
     \centering
    \includegraphics[width=0.8\textwidth]{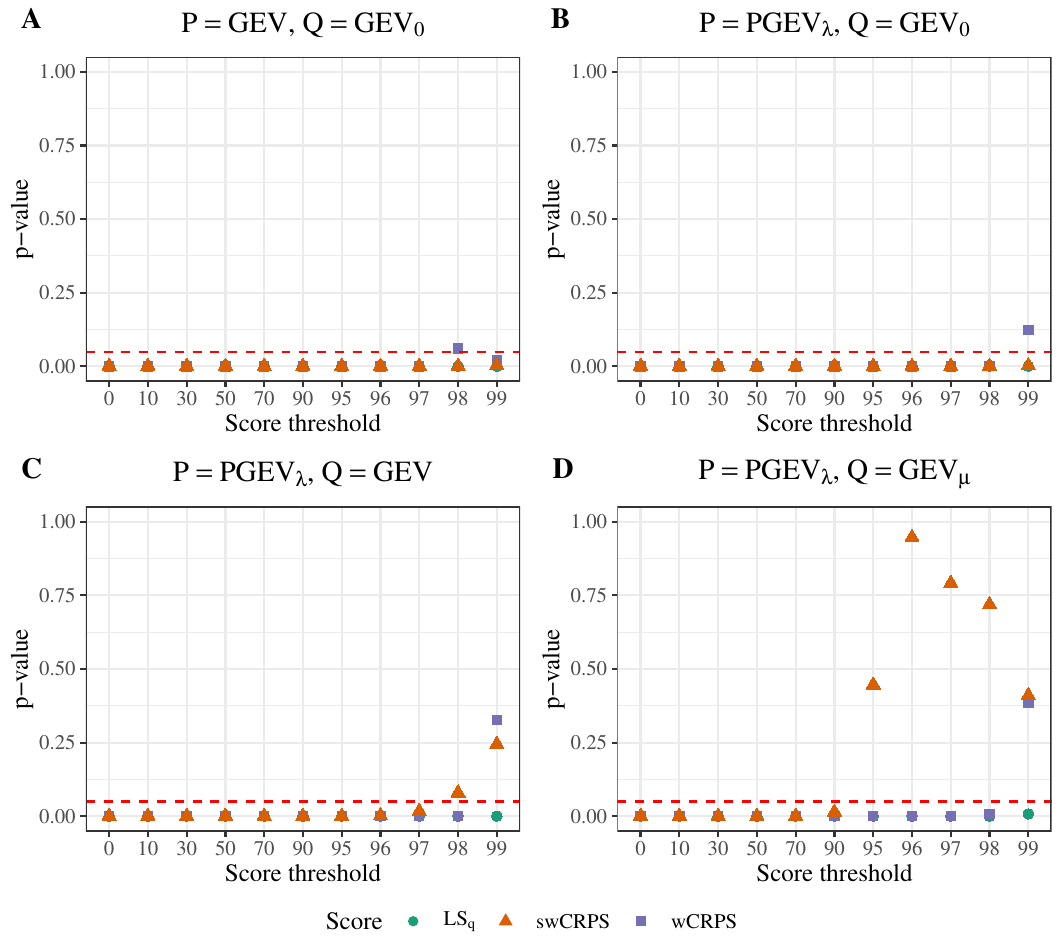}
        \caption{P-values from t-test station-wise score differences $\Delta_i(\p,\Q)$ for predictions $\p$ and $\Q$, where A) $\p=GEV$, $\Q=GEV_{\gamma=0}$, B) $\p=PGEV_\lambda$, $\Q=GEV_{\gamma=0}$ , C) $\p=PGEV_\lambda$, $\Q=GEV$ and D) $\p=PGEV_\lambda$, $\Q=GEV_\mu$. The 0.05 level is marked with a red dashed line. Note that the y-scales are different in the different plots.}
        \label{fig:res:ne:ttest}
\end{figure}

Instead of only considering average  scores, one can, as suggested by \citet{TAILLARDAT2022}, also check for the existence of trends  by permuting  covariates (in this paper temperature) and computing prediction scores for the model which uses the permuted covariates instead of the ordered ones. In the absence of trends, these predictions should be similar to the ones which use the ordered covariates. By plotting the ordered average station scores from the original station data $\mathcal{D}_1=(t_i,y_i)_{i=1}^n$ against the ordered average scores obtained by using the permuted covariates $\mathcal{D}_2=(t_{\pi(i)},y_i)_{i=1}^n$, one can see if the scores from $\mathcal{D}_1$ are larger, which would point at the existence of a trend, or if they distribute evenly around the $45$ degree line, in which case trend is unlikely. For the PGEV model with trend in frequency, Figure \ref{fig:res:ne:permuted} suggests the existence of a trend regardless of which scoring rule is used. Here, the same conclusion is reached regardless of score chosen. In particular, the rankings of the models do not change if we focus only on high values. Whether this is the case is application dependent, and Section~\ref{sec:pm10} shows an example where the conclusions differ when only focusing on high values.

\begin{figure}[tb]
     \centering
 \includegraphics[width=0.9\textwidth]{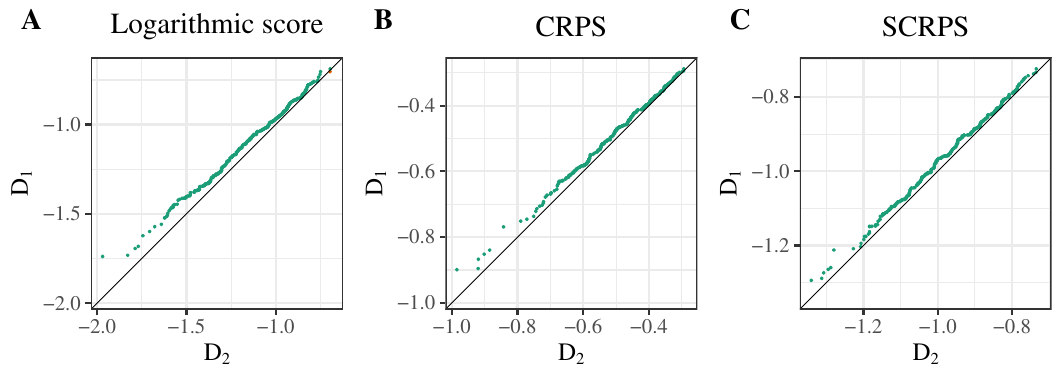}
        \caption{Sorted average station scores for a PGEV model with trend in $\lambda$ plotted for original dataset $\mathcal{D}_1$ against permuted dataset $\mathcal{D}_2$. The colour represents if the points fall above (green) or below (orange) the diagonal. }
        \label{fig:res:ne:permuted}
\end{figure}

\subsection{Particle matter concentration}\label{sec:pm10}

\begin{figure}[tb]
     \centering
  \includegraphics[width=0.9\textwidth]{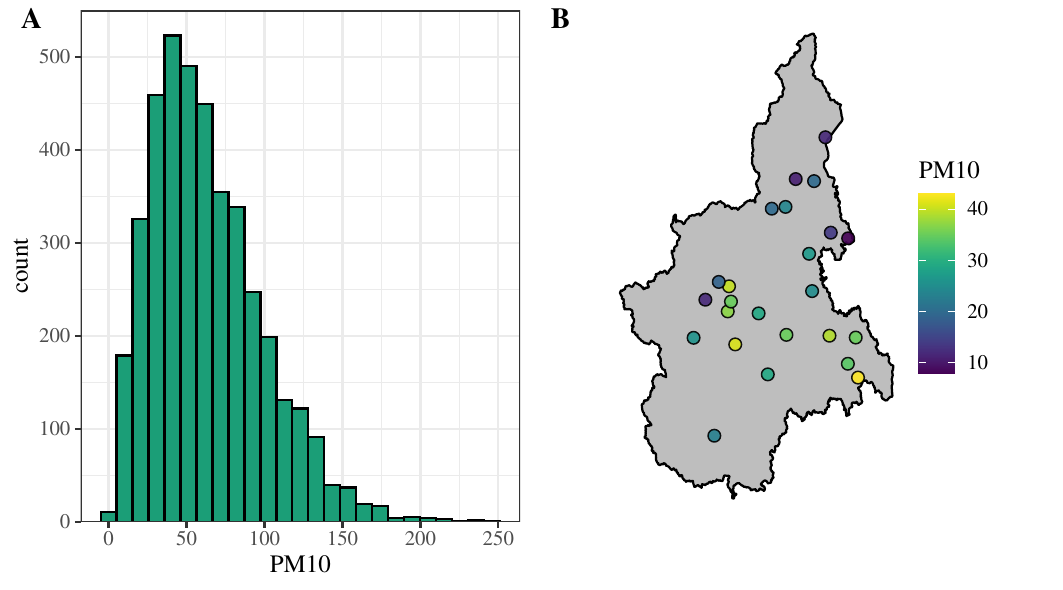}
        \caption{The position of the 24 stations in the Piemonte region along with observations of PM10 concentration on March 13th 2006 (right), and a histogram showing the observations from the 24 stations at the 182 different dates (left).}
        \label{fig:piemonte_data}
\end{figure}

So far, the examples have been focused on extremes, due to the natural interest in predictions above a certain threshold. However, there are many cases where there is a reason to focus on specific regions which are not necessarily extreme. For example, high concentration of particle matter of size less than 10 nm in diameter (PM10) have negative health impact. According to WHO \citep{whoairquality}, the daily PM10 should not exceed $45\text{ }\mu\text{g}/\text{m}^3$. Moreover, based on studies on the connection between long-term PM10 and non-accidental mortality, they found that the long-term level of PM10 should not exceed $15\text{ }\mu\text{g}/\text{m}^3$. This is yet to become a recommendation, but it might be of interest to evaluate how well models can predict PM10, focusing on values above $15\text{ }\mu\text{g}/\text{m}^3$. The assumption is that PM10 concentrations below $15\text{ }\mu\text{g}/\text{m}^3$ have less health implications, and thus the prediction of lower values not as important for the model evaluation. As an illustration, we focus on concentrations above three distinct thresholds; $15\text{ }\mu\text{g}/\text{m}^3$, $30\text{ }\mu\text{g}/\text{m}^3$, and $45\text{ }\mu\text{g}/\text{m}^3$ when comparing predictions of PM10 concentrations in the Piemonte region in Italy. We use the data from \citet{cameletti2013spatio}, containing observations for 182 days from October 2005 to March 2006 at 24 locations. The data is illustrated in Figure~\ref{fig:piemonte_data}.

Denoting the $n=4368$ observations of log-PM10 by $y_i, i=1,\ldots, n$,  \citet{cameletti2013spatio} proposed using a model of the form 
\begin{equation}\label{eq:spacetime}
    y_{i}= \sum_{m=1}^M W_m(s_i,t_i) \beta_m +u(s_i,t_i) + e_i.
\end{equation}
Here, $W_m$ are covariates, $\beta_m$ regression coefficients, and $e_i$ are independent centred Gaussian variables representing measurement noise. Further, $u(s,t)$ is a spatio-temporal Gaussian random field with a separable covariance function.

Since separability is rarely a realistic assumption for spatio-temporal statistical models, we compare the predictions based on this model with those of a model where $u$ is replaced by a Gaussian random field with a non-separable covariance function. Specifically, we use the ``critical diffusion'' model model of \citet{lindgren2023diffusionbased}, which is a Gaussian random field defined as a solution to a stochastic partial-differential equation (SPDE). Both models are fitted to the data using the  code provided in \citep{PiemonteCode}.

To compare the models in terms of predictive performance, a one step ahead cross-validation was preformed by extending the cross-validation function of the R package rSPDE \citep{rspdepackage} to include swCRPS and wCRPS, with the test set at fold $i$ defined as all observations at time $t_{i+1}$, and the test set as all observations up to that time, for $i=1,...,181$. The predictions are compared in terms of CRPS, SCRPS, wCRPS and swCRPS. The results of the best scores are shown in Table \ref{tab:piemonte}. In line with the results found in \citep{lindgren2023diffusionbased}, the CRPS and SCRPS indicate that the non-separable model has better predictive performance than the separable model. However, for the weighted scores, the wCRPS indicates a better performance of the non-separable model, while the swCRPS indicates a better performance of the separable model. Thus, this application is an example of where the choice of scoring rule affects the rankings, and we in particular see that the swCRPS indicates that the separable model might be preferable when focusing on harmful levels of PM10.

\begin{table}
    \centering
    \caption{Mean scores from one-step ahead cross-validation used to compare PM10 predictions in the Piemonte data for  separable and non-separable space-time models as in~\cite{PiemonteCode}. The higher score is represented in bold. The thresholds used for the weighted scores are $15\text{ }\mu\text{g}/\text{m}^3$, $30\text{ }\mu\text{g}/\text{m}^3$, and $45\text{ }\mu\text{g}/\text{m}^3$.}
    \label{tab:piemonte}
    \resizebox{\linewidth}{!}{
    \begin{tabular}{ccccccccccc}
    \toprule
          & \multirow{2}{*}{CRPS} & \multirow{2}{*}{SCRPS} &  \multicolumn{3}{c}{wCRPS} & \multicolumn{3}{c}{swCRPS}\\
         \cmidrule(r){4-6}\cmidrule(r){7-9}
         &&&$15\text{ }\mu\text{g}/\text{m}^3$&$30\text{ }\mu\text{g}/\text{m}^3$&$45\text{ }\mu\text{g}/\text{m}^3$&$15\text{ }\mu\text{g}/\text{m}^3$&$30\text{ }\mu\text{g}/\text{m}^3$&$45\text{ }\mu\text{g}/\text{m}^3$\\
         \hline
     Separable & -0.2243 & -0.5792 & -0.1987 &  -0.1600 & -0.1151 & \textbf{-0.5516} & \textbf{-0.4960}  & \textbf{-0.1452}\\
 Non-separable & \textbf{-0.2232} & \textbf{-0.5750}  & \textbf{-0.1983} & \textbf{-0.1590}& \textbf{-0.1590} & -0.5863  & -1.1672  & -0.2162\\
  \bottomrule
    \end{tabular}}
\end{table}

\section{Conclusion and discussion}\label{sec:conclusion}

Desirable properties such as straightforward computation of scores by Monte Carlo approximation have made the CRPS a popular alternative to the logarithmic score. However, using the CRPS requires sacrificing the local scale invariance of the logarithmic score. 
This can cause different observations, say at different spatial location,  not to be equally important for average scores. To account for this, a common choice in the literature is to use so called skill scores
$(S_n-S_n^{ref})/S_n^{ref}$ or
$(S_n-S_n^{ref})/(S_n^{opt}-S_n^{ref})$,
where $S_n^{ref}$ and $S_n^{opt}$ are the scores for a reference method and the hypothetical optimal score, respectively. However, these scores are often improper, even when they are based on a proper scoring rule \citep{Gneiting2007}. 

\citet{Bolin2022LocalRulesCustom} addressed this issue by introducing the SCRPS, a proper scoring rule which retains most of the desirable properties of the CRPS but in addition also is locally scale invariant.
We extend this construction to the weighted CRPS to obtain the swCRPS and investigate the properties of this score. We show that neither the wCRPS nor the swCRPS are scale invariant, but that the swCRPS is locally tail-scale invariant, meaning that conditioned on exceeding the threshold used in the weighting, the score is locally scale invariant.

Through simulation studies we showed that for a number of extreme value models swCRPS works in a similar way as the SCRPS does for full data sets. In particular, the mean and variability of score differences remain similar when scale parameters are changed, and average scores over random variables with different scales keep the relative error in scale parameter estimation approximately equal.
These properties can be important, for example when evaluating weather models that operate on different scales at different spatial locations, such as temperature, water levels, or rainfall amounts at different measuring stations.

As an example, we used different scores to evaluate five models on annual maximum rainfall data in the Northeastern USA. The comparison led to similar conclusion for all scores, with the best tested model being a PGEV model with a temperature-wise trend in the frequency parameter, meaning that for this case study we come to the same conclusion when only considering the tail of the distribution. This was not the case for the particle matter application, where focusing on high values changed the model ranking.

It should be noted that scale dependence can sometimes be preferred, as mentioned in \citep{Bolin2022LocalRulesCustom}. In such situations, the wCRPS would likely be a better choice than the swCRPS. However, we believe that such situations are rare for typical applications in extremes. Moreover, the $LS_q$ seems to have the best power, as expected, but this score might not be possible to compute for more complicated models. In such cases the swCRPS can be a good alternative.

Issues regarding tail weighted scores have been raised as the forecaster's dilemma \citep{Lerch2017}. The lack of information from observations below a certain threshold can unintentionally draw the score towards the wrong model based on the weight of the models' tail behaviour. This has been addressed by also including the unweighted CRPS by using a weight function of the form $a+bw_u(x)$~\citep{THORARINSDOTTIR2018155}. Another solution might be to simply not choose a ``too high'' threshold $u$.

Weighted scores are not solely relevant for extremes but can be applied whenever predictions within a specific region are of higher importance, as seen for the PM10 concentration model, where a Gaussian spatio-temporal model was used to predict the concentration of PM10 in the Piemonte region of Italy \citep{cameletti2013spatio,lindgren2023diffusionbased}. Moreover, the regions of interest can be defined through other weight functions than the indicator weight function $w_u(y)=1\{y\geq u\}$. The notation of local weight-scale invariance extends the concept of local tail-scale invariance to accommodate for local scale invariance within these types of regions.

\newpage

\appendix

\section{Extreme value models and scoring rules}\label{sec:models}

When considering block maxima data, such as annual maxima of daily precipitation data,  classical extreme value theory \citep[see e.g.][]{Coles2001} shows that the Generalised Extreme Value (GEV) distribution is a suitable model for the data. Different variations of this approach exist, as shown below. 

\subsection{Extreme value theory models}\label{sec:evt:models}
 The GEV distribution has location, scale and shape parameters $\mu,\sigma >0$, and $\gamma$ respectively, and distribution function
\begin{equation*}
    F_{GEV}(x)=\begin{cases}\exp\left(-\exp\left(-\frac{x-\mu}{\sigma}\right)\right)&\text{if }\gamma=0,\\
    \exp\left\{-\left(1+\frac{\gamma}{\sigma}(x-\mu)\right)^{-1/\gamma}\right\}&\text{if }\gamma\neq0.
    \end{cases}
\end{equation*}
In this model, one can introduce trends in the parameters, for example  to model climate change. 
Further, one can use a reparameterisation, PGEV, of the GEV distribution in terms of  parameters $\lambda,\sigma_u,\gamma,u$, with
\begin{equation*}
    \mu_\text{PGEV}=u+\frac{(\lambda^\gamma-1)\sigma_u}{\gamma},\qquad\sigma_{PGEV}=\sigma_u\lambda^\gamma,\qquad\gamma_{PGEV}=\gamma.
\end{equation*}
Here the parameter $\lambda$ describes the frequency of exceedances of a high level $u$, and $\sigma_{PGEV}$ is a scale parameter of the distribution of sizes of excesses of $u$. A reason for this reparameterisation is that it gives a clearer physical understanding of the behaviour of extreme events. The distribution function of the PGEV can be written as
\begin{equation*}
    F_{PGEV}(x)=\exp\left\{-\lambda\left(1+\frac{\gamma}{\sigma_u}(x-u)\right)^{-1/\gamma}\right\}.
\end{equation*}
For details and more information, see~\citet{helgak}. For models without trends in parameters, the PGEV distribution is the same as the GEV distribution. However, the trends in $\lambda$ for the PGEV model and in $\mu$ for the GEV model behave differently. 

When $\p$ is the GEV or PGEV distribution, and $\gamma<1$, the CRPS has a closed form expression \citep[see e.g.][]{Friederichs2012},
\begin{equation*}
    CRPS(\p,y)=
    \begin{cases}
    -(y-\mu+\frac{\sigma}{\gamma})(2F(y)-1)+\frac{\sigma}{\gamma}(2^\gamma \Gamma(1-\gamma)-2\Gamma_l(1-\gamma,-\ln F(y))) & \text{if $\gamma \neq 0$,}\\
    (y-\mu)+2\sigma Ei(\ln F(y)) - \sigma(C-\ln 2) & \text{if $\gamma = 0$,}
    \end{cases}
\end{equation*}
where $\Gamma(a)=\int_0^\infty t^{a-1}e^{-t}dt$ is the gamma function, $\Gamma_l(a,\tau)=\int_0^\tau t^{a-1}e^{-t}dt$ is the lower incomplete gamma function, $Ei(x)=\int_{-\infty}^x e^t/t dt$ is the exponential integral, and $C$ is the Euler-Masceroni constant. For $\gamma\geq 1$, the CRPS does not exist. In Appendix~\ref{sec:appendix:expressions} we derive the corresponding closed-form expressions for the wCRPS.

Through the following reformulation of Eq. \eqref{eq:score:swcrps}, 
\begin{equation*}
\begin{aligned}
    swCRPS(\p,y) &=&& \frac{wCRPS(\p,y)}{E_{\p,\p}[g_w(X,X')]}-\frac{1}{2}\ln(E_{\p,\p}[g_w(X,X')]) - \frac{1}{2},
\end{aligned}
\end{equation*}
we can also evaluate $swCRPS(\p,y)$ for the GEV distribution through the closed form expression for the wCRPS in combination with the expression
\begin{equation*}
\begin{aligned}
    E_{\p,\p}[g_w(X,X')] &= -2\left(q-\mu+\frac{\sigma}{\gamma}\right)F(q)(1-F(q))\\
    &\quad+2\frac{\sigma}{\gamma}\left[2^{\gamma}\Gamma_l(1-\gamma,-2\ln(F(q)))-2\Gamma_l(1-\gamma,-\ln(F(q)))\right].
\end{aligned}
\end{equation*}
See Appendix~\ref{sec:appendix:expressions} for the derivation of these expressions.

\section{Closed form expressions for GEV scores}\label{sec:appendix:expressions}
Let $\p$ be a distribution with CDF $F$ and PDF $f$, and let $w_q$ be the indicator weight function as described in Eq. (\ref{eq:qwf}). Then
\begin{align*}
    E_{\p}&[g_w(X,y)]=\int_{\R} g_w(x,y)f(x)dx\\
    &=(1-w_q(y))\int_q^\infty (x-q)f(x)dx + w_q(y)\left(\int_{-\infty}^q (y-q)f(x)dx+\int_q^\infty|x-y|f(x)dx\right)\\
    &=\int_q^\infty xf(x)dx-q(1-F(q))+
    w_q(y)\left(y(2F(y)-1)-q(2F(q)-1)-2\int_q^yxf(x)dx\right),
\end{align*}
and
\begin{equation}
\begin{aligned}
    E_{\p\p}[g_w(X,X')]&=\int_\R\int_{\R} g_w(x,y)f(x)dxf(y)dy\\
    &=4\int_q^\infty xF(x)f(x)dx-2\int_q^\infty xf(x)dx-2qF(q)(1-F(q)).
\end{aligned}\label{eq:expected:difference:WX:appendix}
\end{equation}
Using the formulation in Eq. \eqref{eq:wCRPS}, we have
\begin{equation}\label{eq:wCRPS:appendix}
\begin{aligned}
w&CRPS(\p,y)=\frac{1}{2}E_{\p\p}F[g_w(X,X')]-E_{\p}[g_w(X,y)]\\
     &=2\int_q^\infty xF(x)f(x)dx-\int_q^\infty xf(x)dx-qF(q)(1-F(q)) -\int_q^\infty xf(x)dx+q(1-F(q))\\
     &\quad-w_q(y)\left(y(2F(y)-1)-q(2F(q)-1)-2\int_q^yxf(x)dx\right)\\
     &=-2\int_{q\vee y}^\infty xf(x)dx+2\int_q^\infty xF(x)f(x)dx-(q\vee y)(2F(q\vee y))-1)+qF(q)^2.
\end{aligned}
\end{equation}
The expressions in Eq. (\ref{eq:expected:difference:WX:appendix}) and (\ref{eq:wCRPS:appendix}) can be rewritten in terms of quantiles of the distribution $F$, since
$\int_{a}^\infty xf(x)dx=\int_{F(a)}^1 F^{-1}(t)dt$ and 
$\int_a^\infty xF(x)f(x)dx=\int_{F(a)}^1 F^{-1}(t)tdt$.
This gives
\begin{equation}\label{eq:expected:difference:WX:appendix:quantile}
    E_{\p\p}[g_w(X,X')]=4\int_{F(q)}^1 F^{-1}(t)tdt-2\int_{F(q)}^1 F^{-1}(t)dt-2qF(q)(1-F(q)).
\end{equation}
and
\begin{equation*}
\begin{aligned}
    wCRPS(\p,y)&=-2\int_{F(q\vee y)}^1 F^{-1}(t)dt+2\int_{F(q)}^1 F^{-1}(t)tdt-(q\vee y)(2F(q\vee y))-1)+qF(q)^2.
\end{aligned}
\end{equation*}
Let us now use these results to evaluate the scoring rules for the GEV distribution. We first consider the case $\gamma=0$, corresponding to the Gumbel distribution, and then the case $\gamma\neq 0$.

\subsection{Scoring rules for the Gumbel distribution}
When $\gamma = 0$, the GEV distribution has CDF $F(x) = \exp(-\exp(-(x - \mu)/\sigma))$ with inverse $F^{-1}(x) = \mu - \sigma\ln(-\ln x))$. Therefore,
\begin{align*}
    \int_{F(q\vee y)}^1 F^{-1}(t)dt&= \mu(1-F(q\vee y))-\sigma\int_{F(q\vee y)}^1 \ln(-\ln t)dt,\\
    \int_{F(q)}^1 F^{-1}(t)tdt &= \frac{1}{2}\mu(1-F(q)^2)-\sigma\int_{F(q)}^1 t\ln(-\ln t)dt,
\end{align*}
which gives
\begin{equation*}
\begin{aligned}
    wCRPS(\p,y)&=&&((q\vee y)-\mu)-\sigma(C-\ln(2))\\
    &&&-\sigma\left(Ei\left(-2\exp\left(-\frac{q-\mu}{\sigma}\right)\right)-2Ei\left(-\exp\left(-\frac{q\vee y-\mu}{\sigma}\right)\right)\right),
\end{aligned}
\end{equation*}
and
\begin{equation*}
\begin{aligned}
    E_{\p\p}[g_w(X,X')]&= 2\sigma\ln(2)-2(q-\mu)(1-F(q))\\
    &\quad-2\sigma\left(Ei\left(-2\exp\left(\frac{q-\mu}{\sigma}\right)\right)-Ei\left(-\exp\left(\frac{q-\mu}{\sigma}\right)\right)\right).
    \end{aligned}
\end{equation*}
    
\subsection{Scoring rules for the GEV distribution with non-zero shape parameter}

When $\gamma \neq 0$, the GEV distribution has CDF $F(x) = \exp\left[-\left(1+\frac{\gamma}{\sigma}(x-\mu)\right)^{-1/\gamma}\right]$ with inverse $F^{-1}(x) = \mu -\frac{\sigma}{\gamma}\left(1-(-\ln x)^{-\gamma}\right)$. Therefore,
\begin{equation*}
\begin{aligned}
    wCRPS(F,y)&=&&\left(q-\mu+\frac{\sigma}{\gamma}\right)F(q)^2-((y\vee q)-\mu+\frac{\sigma}{\gamma})(1-2F(y\vee q))\\
    &&&-\frac{\sigma}{\gamma}\left[2^\gamma\Gamma_l(1-\gamma,-2\ln(F(q)))-2\Gamma_l(1-\gamma,-\ln(F(y\vee q)))\right],
\end{aligned}
\end{equation*}

and
\begin{equation*}
\begin{aligned}
    E_{\p\p}[g_w(X,X')]&=&&-2\left(q-\mu+\frac{\sigma}{\gamma}\right)F(q)(1-F(q))\\
    &&&+2\frac{\sigma}{\gamma}\left[2^{\gamma}\Gamma_l(1-\gamma,-2\ln(F(q)))-2\Gamma_l(1-\gamma,-\ln(F(q)))\right],
\end{aligned}
\end{equation*}

where $\Gamma_l(a,\tau)=\int_0^\tau t^{a-1}e^{-t}dt$.

\section{Proofs}\label{sec:appendix:proofs}

\begin{proof}[Proof of Proposition \ref{prop:wcrps:kernelscore}]
Let $g_w(x,y):=|\int_{y}^x w(t)dt|$, where $w(x)$ is a non-negative function. Then $g_w:\Omega\times\Omega\to\R$ is a symmetric real-valued function. Moreover, 

\begin{equation*}
    \begin{aligned}
        \sum_{i=1}^n\sum_{j=1}^n a_ia_jg_w(x_i,x_j)\leq 0
    \end{aligned}
\end{equation*}
for all $n$, all $a_1,...,a_n\in\R$ s.t. $\sum_{i=1}^n a_i=0$ and all $x_1,...,x_n\in\Omega$, since $g(x,y)=|x-y|$ is negative definite. Hence, the kernel $g_w$ is negative definite and the weighted CRPS,
\begin{equation*}
    wCRPS(\p,y)=\frac{1}{2}E_{\p,\p}[g_w(X,X')]-E_{\p}[g_w(X,y)],
\end{equation*}
is indeed a kernel score. 
\end{proof}

\begin{proof}[Proof of Proposition \ref{prop:wCRPS:scale:function}]

Note that for any $\Q$ for which Assumption~\ref{assumption1} holds, the assumption also holds for $\Q^{w_u}$, where $u\in \text{supp}(\Q)$, and $w_u$ is the indicator weight function.\\

We start proving (i). Assume that $\Q_{\thetab}$ is a probability measure satisfying Assumption \ref{assumption1}. Let $S(\Q_{\thetab},y)=\frac{1}{2}E_{\Q_{\thetab},\Q_{\thetab}}[g_w(X,Y)]-E_{\Q_{\thetab}}[g_w(X,y)]$ denote the wCRPS where $q_\theta$ is the density of $\Q_{\thetab}$. 
By Tailor expansion around $\thetab$, we have that
\begin{equation}
    S(\Q_{\thetab+t\sigma \rb},\Q_{\thetab})=S(\Q_{\thetab},\Q_{\thetab})+t\sigma \rb^T\nabla_{\thetab} S(\Q_{\thetab},\Q_{\thetab})+\frac{1}{2}t^2\sigma^2\rb^T\nabla_{\thetab}^2S(\Q_{\thetab},\Q_{\thetab})\rb+o(t^2).  \label{eq:proof:taylor}
\end{equation}

From Assumption \ref{assumption1}, $s(\theta)=\nabla_\theta^2S(\Q_{\thetab},\Q_{\thetab})|_{\Q=\Q_{\thetab}}$ exists and is continuous and we have $\nabla_\theta S(\Q_{\thetab},\Q_{\thetab})=0$ since $S$ is proper. Hence,
\begin{equation*}
\begin{aligned}
    S(\Q_{\thetab},\Q_{\thetab})-S(\Q_{\thetab+t\sigma \rb},\Q_{\thetab})&=-\frac{1}{2}t^2\sigma^2\rb^T\nabla_{\thetab}^2S(\Q_{\thetab},\Q_{\thetab})\rb+o(t^2)\\
    &=\frac{1}{4}t^2\sigma^2\rb^T\nabla_{\thetab}^2E_{\Q_{\thetab},\Q_{\thetab}}[g_w(X,Y)]\rb+o(t^2). \label{eq:proof:taylor2}
\end{aligned}
\end{equation*}
We can now follow the steps in the proof of Lemma 1 in \citet{Bolin2022LocalRulesCustom}, where we replace their kernel $g_c(x,y)$ with our weighted kernel $g_w(x,y)$, since the only thing required in the proof is that the kernel is a positive negative-definite kernel and that $g_w(x,y)\leq g(x,y)$, both of which hold whenever $|w(t)|\leq 1$ for all $t\in\R$. This holds certainly for our indicator weight function $w_q$ but even for other choices of $w$. 
This results in
\begin{equation}
  \nabla_{\thetab}E_{\Q_{\thetab},\Q}[g_w(X,Y)]=-\sigma^{-1}E_{\Q,\Q}[g_w(\sigma X+\mu, Y)v(X)],\label{eq:proof:nabla} 
\end{equation}
and
\begin{equation}
  \nabla_{\thetab}^2E_{\Q_{\thetab},\Q_{\thetab}}[g_w(X,Y)]=\sigma^{-2}E_{\Q,\Q}[g_w(\sigma X+\mu,\sigma Y+\mu)H(X,Y)],  \label{eq:proof:nabla:2}
\end{equation}
where $v(X)$ is a vector and $H(X,Y)$ is a $2\times 2$ matrix, both independent of $\thetab$, and
\begin{equation*}
\begin{aligned}
    g_w(\sigma x+\mu,\sigma y+\mu)&=\left|\int_{\sigma y+\mu}^{\sigma x+\mu} w(t)dt\right|
    =\sigma\left|\int_{y}^{x} w(\sigma \tau+\mu)d\tau\right|.
\end{aligned}
\end{equation*}
For $w=w_q$, we have
\begin{equation}\label{eq:swcrps:g:cases}
\begin{aligned}
    g_w(\sigma x+\mu,\sigma y+\mu)&=\sigma\left|\int_{y}^{x} 1\left\{\tau\geq \frac{q-\mu}{\sigma}\right\}d\tau\right|=
    \begin{cases}
    \sigma|x-y| & \text{if } x,y\geq \frac{q-\mu}{\sigma},\\
    \sigma|y-\frac{q-\mu}{\sigma}| & \text{if } y\geq \frac{q-\mu}{\sigma}\geq x,\\
    \sigma|x-\frac{q-\mu}{\sigma}| & \text{if } x\geq \frac{q-\mu}{\sigma}\geq y,\\
    0& \text{otherwise.}
    \end{cases}
\end{aligned}
\end{equation}
Since we cannot let $w$ depend on $\thetab=(\mu,\sigma)$, we cannot choose any $w$ such that 
\begin{equation}
g_w(\sigma x+\mu,\sigma y+\mu)=\sigma g_w(x, y).\label{eq:gwgw}
\end{equation}
However, letting $k=\frac{q-\mu}{\sigma}$, we have that 
$g_w(\sigma x+\mu,\sigma y+\mu)=\sigma|x-y|$
for all $x,y\geq k$.

Then the wCRPS has scale function
\begin{equation}
    s(\Q_{\thetab})=\sigma^{-1} E_{\Q,\Q}[H_\Q(X,Y)h(X,Y)],
\end{equation}
where $H_\Q(X,Y)$ is a $2\times2$ matrix independent of $\theta$ and $h(X,Y)$ is a function dependent on $\thetab$.
Furthermore, there exists a constant $k<\infty$, such that $h(X,Y)$ is independent of $\theta$ for all $X,Y\geq k$.
From this scale function we directly see that the wCRPS is neither locally scale invariant nor locally tail-scale invariant. 

We now prove (ii).
Let $S$ be the swCRPS defined in Eq.~\eqref{eq:score:swcrps}. As in the proof of statement (i), we perform the Taylor expansion shown in Eq.~\eqref{eq:proof:taylor}.  Since $S$ is proper, $\nabla_\theta S(\Q_{\thetab},\Q_{\thetab})=0$ and we only need to consider the term $\nabla_\theta^2S(\Q_{\thetab},\Q)|_{\Q=\Q_{\thetab}}$ that exists and is continuous. 

For simplified notation, let $E_{P,Q}=E_{\p_{\thetab},\Q}[g(X,Y)]$, $E_{\dot{P},Q}=\nabla_{\thetab}E_{\p_{\thetab},\Q}[g(X,Y)]$, $E_{Q,\dot{P}}=\nabla_{\thetab}E_{\Q,\p_{\thetab}}[g(X,Y)]$, and $E_{\ddot{P},Q}=\nabla^2_{\thetab}E_{\p_{\thetab},Q}[g(X,Y)]$. One can show that

\begin{equation*}
\begin{aligned}
    \nabla_\theta^2\log(E_{P,P})=&-\frac{2E_{\dot{P},P}E_{\dot{P},P}^T}{E_{P,P}^2}+\frac{E_{\dot{P},\dot{P}}}{E_{P,P}}+\frac{E_{\ddot{P},P}}{E_{P,P}},\\
    \nabla_\theta^2\frac{E_{P,Q}}{E_{P,P}}=&-\frac{2E_{\dot{P},Q}E_{\dot{P},P}^T}{E_{P,P}^2}-\frac{2E_{P,Q}E_{\dot{P},P}^T}{E_{P,P}^2}-\frac{2E_{P,Q}E_{\dot{P},\dot{P}}^T}{E_{P,P}^2}\\
    &+\frac{2^3E_{\dot{P},P}E_{\dot{P},P}^T}{E_{P,P}^3}+\frac{E_{\ddot{P},P}}{E_{P,P}}-\frac{2E_{P,Q}E_{\ddot{P},P}^T}{E_{P,P}^2}.
\end{aligned}
\end{equation*}
Evaluating these together at $Q=P$ yields
\begin{equation*}
    \nabla_\theta^2S(\p_{\thetab},\Q)|_{\Q=\p_{\thetab}} = \frac{1}{E_{P,P}}E_{\dot{P},\dot{P}}-\frac{2}{E^2_{P,P}}E_{\dot{P},P}E_{P,\dot{P}}^T.
\end{equation*}
Inserting \eqref{eq:proof:nabla} and \eqref{eq:proof:nabla:2} instead of $E_{\dot{P},P}$ and $E_{\dot{P},\dot{P}}$ into the equation above results in 
\begin{equation*}
\begin{aligned}
    s(\Q_{\thetab}) =& -\frac{1}{2}\nabla_\theta^2S(\Q_{\thetab},\Q)|_{\Q=\Q_{\thetab}}\\
    =&-\frac{\sigma^{-2}}{2}\left(\frac{E_{\Q,\Q}[g_w(\sigma X+\mu,\sigma Y+\mu)H(X,Y)]}{E_{\Q,\Q}[g_w(\sigma X+\mu,\sigma Y+\mu)]}\right.\\
    &\left.-\frac{2E_{\Q,\Q}[g_w(\sigma X+\mu, \sigma Y+\mu)v(X)]E_{\Q,\Q}[g_w(\sigma X+\mu, \sigma Y+\mu)v(X)]^T}{E_{\Q,\Q}[g_w(\sigma X+\mu,\sigma Y+\mu)]^2}\right).
    \end{aligned}
\end{equation*}
and for $u\geq k$,
\begin{equation*}
    s(\Q_{\thetab}^{w_u}) = -\frac{\sigma^{-2}}{2}\left(\frac{\sigma E_{\Q^{w_u},\Q^{w_u}}[|X-Y|H(X,Y)]}{\sigma E_{\Q^{w_u},\Q^{w_u}}[|X-Y|]}-\frac{2\sigma^2E_{\Q^{w_u},\Q^{w_u}}[|X-Y|v(X)]E_{\Q^{w_u},\Q^{w_u}}[|X-Y|v(X)]^T}{\sigma^2 E_{\Q^{w_u},\Q^{w_u}}[|X-Y|]^2}\right),
\end{equation*}
i.e.
\begin{equation}
    s(\Q_{\thetab}^{w_u})=s(\Q_{(0,1)}^{w_u})
\end{equation}
for $u\geq k$, showing that the swCRPS is locally tail-scale invariant.

\end{proof}

\begin{proof}[Proof of Proposition \ref{prop:rCRPS}]
\begin{enumerate}[(i)]
    \item According to \citet{Bolin2022LocalRulesCustom}, the scale function of rCRPS with kernel function 
\begin{equation}\label{eq:rcprs:g}
    g_c(x,y)=\begin{cases}|x-y|,&|x-y|<c\\0,&\text{otherwise}\end{cases}
\end{equation}
is 
\[s(\Q_{\theta}) = \sigma^{-1}E_{\Q,\Q}\left[H_{\Q}(X,Y)1\left(|X-Y|<\frac{c}{\sigma}\right)\right],\]
where $H_{\Q}(X,Y)$ is a $2\times2$ positive semidefinite matrix independent of $\theta$.
Choose $u=\frac{c}{2\sigma}$. Then, for $|x|<u, |y|<u$,
\[|x-y|\leq|x|+|y|<\frac{c}{2\sigma},\]
and
\[s(\Q_{\theta}^{w^u}) = \sigma^{-1}E_{\Q^{w^u},\Q^{w^u}}\left[H_{\Q^{w^u}}(X,Y)\right].\]
Therefore, rCRPS is not locally weight-scale invariant with respect to $w^u(x)=1\{|x|<u\}$
\item Following the proof of Proposition \ref{prop:wCRPS:scale:function}, using the fact of Eq. \eqref{eq:rcprs:g} instead of Eq. \eqref{eq:swcrps:g:cases}, yields that the rSCRPS is locally weight-scale invariant for weight function $w(x)=1\{|x|<c/2\}$.
\end{enumerate}
\end{proof}

\begin{proof}[Proof of Proposition \ref{prop:lsq}]

For $u\geq q$, the censored likelihood score, $LS_q$, and the logarithmic score, $LS$, have the same conditional expectation, since
\begin{equation*}
\begin{aligned}
S_u(\p,\Q) &= E_\Q[LS_q(\p_u,Y)|Y>u]
=E_\Q[1\{Y\leq q\}\log(F(q))+1\{Y > q\}\log(f(Y))|Y>u]\\
&=E_\Q[\log(f(Y))|Y>u]
=E_\Q[LS(\p_u,Y)|Y>u].
\end{aligned}
\end{equation*}
The logarithmic score is locally scale invariant~\citep{Bolin2022LocalRulesCustom}, and thus also locally tail-scale invariant according to Remark \ref{remark:tail:scale}. Therefore, the conditional likelihood is also locally tail-scale invariant.
\end{proof}

\begin{proof}[Proof of Proposition \ref{prop:general_cencoredlikelohood}]
First note that
\begin{equation}
    \begin{aligned}
        \text{CLog}_w(\p,y)&=w(y)\log\left(\int_\R p(t)w(t)dt\right)+(1-w(y))\log\left(1-\int_\R p(t)w(t)dt\right)+w(y)\log(p_w(y))\\
        &=w(y)\log\left(p(y)\right)+(1-w(y))\log\left(1-\int_\R p(t)w(t)dt\right)
    \end{aligned}
\end{equation}
Since conditioned on $\tilde{w}(y)=1\{w(y)==1\}$, the score becomes the scale invariant logarithmic score, the $\text{CLog}_w$ is locally weight-scale invariant with respect to $\tilde{w}$.
\end{proof}

\begin{proof}[Proof of Proposition \ref{prop:combined-scores}]
\begin{enumerate}[(i)]
    \item Consider the score
    \begin{equation*}
        owS(\p,y;w)=w(y)S_0(\p^w,y)
    \end{equation*}
    Note that 
    \begin{equation*}
        \begin{aligned}
    owS(\p^w,\Q^w;w)=E[w(Y)S_0(\p^w,Y)|w(Y)>0]=E[S_0(\p^w,Y)|w(Y)>0]=S_0(\p^w,\Q^w)
        \end{aligned}
    \end{equation*}
    so if $S_0$ is locally weight-scale invariant with respect to $w$, then $owS$ is also locally weight-scale invariant with respect to $w$.
    
    \item Consider the score $S=S_1+S_2$. The weight-scale function of $S$ is the sum of the weight-scale functions of $S_1$ and $S_2$. If $S_1$ and $S_2$ locally weight-scale invariant with respect to $w$, their weight-scale functions fulfil the requirements for local scale invariance, and hence, the weight-scale function of $S$ does too. Therefore, $S$ is locally weight-scale invariant.
\end{enumerate}
\end{proof}

\section{Estimated flood densities}\label{sec:app:flood:pdf}

\begin{figure}[t]
    \centering
    \includegraphics[width=0.95\textwidth]{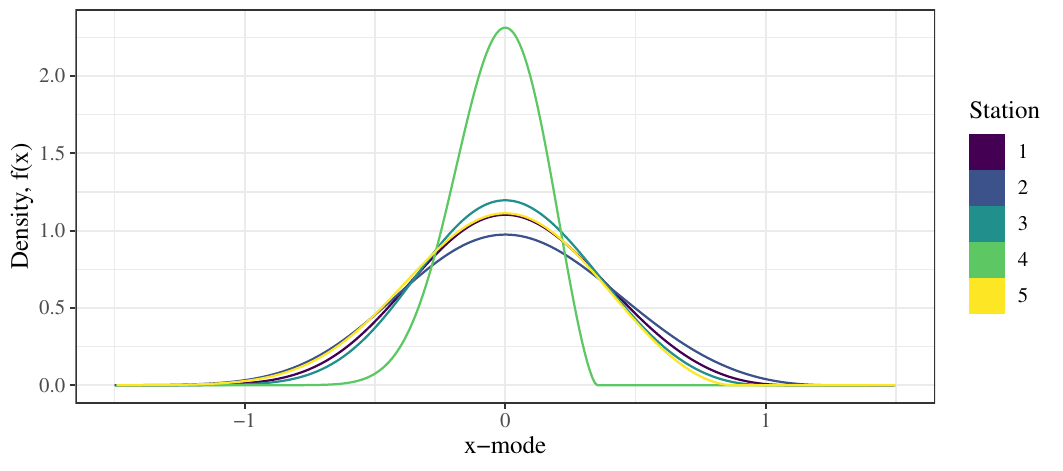}
    \caption{Density of estimated GEV distributions for flood stations as described in Table \ref{tab:great:lakes:gev:fit}, after shifting the x-axis by the estimated mode of each station.}
    \label{fig:pdf:gev:lakes}
\end{figure}

The annual maximum water levels were modelled with a stationary GEV distribution and the estimated parameters shown in Table \ref{tab:great:lakes:gev:fit}. Stations Michigan Huron (station 2) and Lake Superior (station 4) were used for simulations since those stations suggested stationarity and since the scale parameter of Lake Michigan Huron was almost twice the scale parameter of Lake Superior. All stations had negative shape parameter. In Figure \ref{fig:pdf:gev:lakes} the estimated density function is compared at all five stations, where the x-axis has been shifted by each respective mode for readability.


\bibliography{bibfin}

\end{document}